\documentclass[journal, onecolumn, draftcls,12pt]{IEEEtran}
\usepackage{amsfonts}
\usepackage{amssymb}
\usepackage{amsmath}
\usepackage{amsthm}

\usepackage{algorithm}
\usepackage{algpseudocode}
\usepackage {graphicx}
\usepackage{bigstrut} 
\usepackage{multirow}
\usepackage{url}
\usepackage{cite}
\usepackage{framed}
\usepackage{subfigure}

\newtheorem{thm}{Theorem}
\newtheorem{lem}{Lemma}
\newtheorem{prop}{Proposition}

\DeclareMathOperator*{\esupa}{arg\,max}
\DeclareMathOperator*{\esupmin}{min}
\DeclareMathOperator*{\esupmax}{max}

\def\x{{\mathbf x}}
\def\X{{\mathbf X}}

\def\b{{\mathbf b}}
\def\B{{\mathbf B}}
\def\w{{\mathbf w}}
\def\W{{\mathbf W}}

\def\V{{\mathbf V}}

\def\e{{\mathbf e}}

\def\0{{\mathbf 0}}
\def\A{{\mathbf A}}

\def\I{{\mathbf I}}

\def\T{{\mathcal T}}
\def\TColon{{\T,\,:}}
\def\THatColon{{\THat,\,:}}
\def\THatDeltaColon{{\tDeltaHat,\,:}}
\def\THatiColon{{\TiHat,\,:}}

\def\GammaColon{{\Gamma,\,:}}
\def\GammaCompColon{{\Gamma^c,\,:}}

\def\THatCompColon{{\THat^c,\,:}}
\def\THatiCompColon{{\THat_i^c,\,:}}

\def\R{{\mathbb R}}
\def\E{{\mathbb E}}
\def\THat{{\hat\T}}

\def\XHat{{\hat\X}}

\def\TiHat{{\hat{\T_i}}}
\def\tDeltaHat{{\hat\T\hspace{-0.11cm}_\Delta}}
\def\deltaRK{\delta_{R+K}}

\newcommand{\vecNorm}[1]{\left\lVert#1\right\rVert}
\newcommand{\vecLTwoNorm}[1]{{\left\lVert#1\right\rVert}_2}

\begin{document}
\title{Fusion of Sparse Reconstruction Algorithms for Multiple Measurement Vectors}
\author{\IEEEauthorblockN{Deepa K G, Sooraj K.~Ambat, and {K.V.S.~Hari}} \\ 
	\thanks{Deepa K G, Sooraj K.~Ambat, and K.V.S.~Hari are with Statistical Signal Processing Lab, Department of Electrical Communication Engineering, Indian Institute of Science, Bangalore, 560012, India. (email:\{deepa,sooraj,hari\}@ece.iisc.ernet.in)}
}
	\maketitle
	\begin{abstract}
		We consider the recovery of sparse signals that share a common support from multiple measurement vectors. The performance of several algorithms developed for this task depends on parameters like dimension of the sparse signal, dimension of measurement vector, sparsity level, measurement noise. We propose a fusion framework, where several multiple measurement vector reconstruction algorithms participate and the final signal estimate is obtained by combining the signal estimates of the participating algorithms. We present the conditions for achieving a better reconstruction performance than the participating algorithms. Numerical simulations demonstrate that the proposed fusion algorithm often performs better than the participating algorithms.
		\end{abstract}
\section{Introduction}
\label{sec:intro}
Consider the standard Compressed Sensing (CS) measurement setup where a $K$-sparse signal $\x \in \R^{N\times 1}$ is acquired through $M$ linear measurements via 
\begin{align}
\label{eqn:CS_Std_Eqn}
\b = \A \x + \w,
\end{align}
where $\A \in \R^{M\times N}$ denotes the measurement matrix, $\b \in \R^{M\times 1}$ represents the measurement vector, and $\w \in \R^{M \times 1}$ denotes the additive measurement noise present in the system. The reconstruction problem, estimating $\x$ from \eqref{eqn:CS_Std_Eqn} using $\A$ and $\b$, is known as Single Measurement Vector (SMV) problem. In this work, we consider the Multiple Measurement Vector (MMV) problem \cite{Cotter2005SparseSolution} where we have $L$ measurements:  $\mathbf{b}^{(1)}=\mathbf{A}\mathbf{x}^{(1)} + \w^{(1)}$, $\mathbf{b}^{(2)}=\mathbf{A}\mathbf{x}^{(2)} + \w^{(2)}$, $\cdots$, $\mathbf{b}^{(L)}=\mathbf{A}\mathbf{x}^{(L)} + \w^{(L)}$. The vectors $\{\mathbf{x}^{(l)}\}_{l=1}^{L}$ are assumed to have a common sparse support-set. The problem is to estimate $\mathbf{x}^{(l)} \, (l=1,2,\dots, L)$. Instead of recovering the $L$ signals individually, the attempt in the MMV problem is to simultaneously recover all the $L$ signals. MMV problem arises in many applications such as the neuromagnetic inverse problem in Magnetoencephalography (a modality for imaging the brain) \cite{Gorodnitsky1995Neuromagnetic,Gorodnitsky1997SparseSignal}, array processing \cite{Malioutov2005ASparseSignal}, non-parametric spectrum analysis of time series \cite{Stoica1997Introduction}, and equalization of sparse communication channels \cite{Cotter2002SparseChannel}.
\par
Recently many algorithms have been proposed to recover signal vectors with a common sparse support. Some among them are algorithms based on diversity minimization methods like $\ell _{2,1}$ minimization \cite{Tropp2006Algorithms}, and M-FOCUSS \cite{Cotter2005SparseSolution}, greedy methods like M-OMP and M-ORMP \cite{Cotter2005SparseSolution}, and Bayesian methods like MSBL~\cite{Wipf2007AnEmpirical} and T-MSBL~\cite{Zhang11SparseSignal}. 
\par
However it has been observed that the performance of many algorithms depends on many parameters like the dimension of the measurement vector, the sparsity level, the statistical distribution of the non-zero elements of the signal, the measurement noise power etc.~\cite{Zhang11SparseSignal}. Thus it becomes difficult to choose the best sparse reconstruction algorithm without {\em a priori} knowledge about these parameters. 
\par
Suppose we have the sparse signal estimates given by various algorithms. It may be possible to merge these estimates to form a more accurate estimate of the original. This idea of fusion of multiple estimators has been proposed in the context of signal denoising in \cite{Elad2009APlurality} where fusion was performed by plain averaging. Recently, Ambat {\em et al.}~\cite{facs,Ambat2013ICASSP,Ambat2014ACommitteMachine,Ambat2014ProgressiveFusion,Ambat2012FuGP_EUSIPCO} proposed fusion of the estimates of sparse reconstruction algorithms to improve the sparse signal reconstruction performance of SMV problem.
\par
	In this paper, we propose a framework which uses several MMV reconstruction algorithms and combines their sparse signal support estimates to determine the final signal estimate. We refer to this scheme as \textit{MMV-Fusion of Algorithms for Compressed Sensing} (MMV-FACS). We present an upper bound on the reconstruction error by MMV-FACS. We also present a sufficient condition for achieving a better reconstruction performance than any participating algorithm. By Monte-Carlo simulations we show that fusion of viable algorithms leads to improved reconstruction performance for the MMV problem.
	
\subsection*{Notations:} Matrices and vectors are denoted by bold upper case and bold lower case letters respectively. Sets are represented by upper case Greek alphabets and calligraphic letters. $\mathbf{A}_\mathcal{T}$ denotes the column sub-matrix of $\mathbf{A}$ where the indices of the columns are the elements of the set $\mathcal{T}$. $\mathbf{X}_{\mathcal{T},\,:}$  denotes the sub-matrix formed by those rows of $\mathbf{X}$ whose indices are listed in the set $\mathcal{T}$. $\mathbf{X}^{K}$ is the matrix obtained from $\mathbf{X}$ by keeping its $K$ rows with the largest $\ell_2$-norm and by setting all other rows to zero, breaking ties lexicographically. \texttt{supp}($\mathbf{X}$) denotes the set of indices of non-zero rows of $\mathbf{X}$. For a matrix $\mathbf{X}$, $\mathbf{x}^{(l)}$ denotes the $\ell^\text{th}$ column vector of $\mathbf{X}$. ${\hat{\mathbf{X}}}_i$ denotes the reconstructed matrix by the $i^\text{th}$ participating algorithm. The complement of the set $\mathcal{T}$ with respect to the set $\{1,2, \ldots , N\}$ is denoted by $\mathcal{T}^c$. For two sets $\mathcal{T}_1$ and $\mathcal{T}_2$, $\mathcal{T}_1  \setminus \mathcal{T}_2 = \mathcal{T}_1 \cap \mathcal{T}_2^c$ denotes the set difference. $|\mathcal{T}|$ denotes the cardinality of set $\mathcal{T}$. $\mathbf{A}^ \dagger$ and $\mathbf{A}^T$ denote the pseudo-inverse and transpose of matrix $\mathbf{A}$, respectively. The $(p,q)$ mixed norm of the matrix $\mathbf{X}$ is defined as 
			\begin{equation*}
			\left \|\mathbf{X}\right \|_{(p,q)}=\left(\sum _{i} \left \|\mathbf{X}_{i,:}\right \|_p^q \right)^{1/q}
			\end{equation*}
			The Frobenius norm of matrix $\mathbf{A}$ is denoted as $\left \|\mathbf{A}\right \|_F$. 
		
\section{PROBLEM FORMULATION} 
The MMV problem involves solving the following $L$ under-determined systems of linear equations
\begin{align}
	\label{Eqn:eqmmv1}
\mathbf{b}^{(l)} = \mathbf{Ax}^{(l)} + \mathbf{w}^{(l)}, \quad l=1,2,3, \dots ,L 
\end{align}
where  $ \mathbf{A} \in \mathbb{R}^{M\times N} \, (M \ll N)$ represents the measurement matrix, $\mathbf{b}^{(l)} \in \mathbb{R}^{M\times 1}$  represents the $l^\text{th}$ measurement vector, and $\mathbf{x}^{(l)} \in \mathbb{R}^{N\times 1}$ denotes the corresponding $K$-sparse source vector. That is,  $|\texttt{supp}(\mathbf{x}^{(l)})| \leq K$ and $\x^{(l)}$ share a common support-set for $l=1,2,\dots,L$. $\mathbf{w}^{(l)} \in \mathbb{R}^{M\times 1}$  represents the additive measurement noise. We can rewrite \eqref{Eqn:eqmmv1} as
\begin{align}
	\label{Eqn:mmvmatrix}
\mathbf{B} = \mathbf{AX}+\mathbf{W} 
\end{align}
where $\mathbf{X}=[\mathbf{x}^{(1)}, \mathbf{x}^{(2)}, \dots , \mathbf{x}^{(L)}]$, $\mathbf{W}=[\mathbf{w}^{(1)}, \mathbf{w}^{(2)}, \dots , \mathbf{w}^{(L)}]$, and $\mathbf{B}=[\mathbf{b}^{(1)}, \mathbf{b}^{(2)}, \dots , \mathbf{b}^{(L)}]$. {For a matrix $\mathbf{X}$, we define
	\begin{align*}
	\texttt{supp}(\mathbf{X})= \bigcup_{i=1}^L \texttt{supp}(\mathbf{x}^i).
	\end{align*}
	In \eqref{Eqn:mmvmatrix}, we assume that $\X$ is jointly $K$-sparse. That is, $|\texttt{supp}(\mathbf{X})| \leq K $. There are at most $K$ rows in $\mathbf{X}$ that contain non-zero elements. We assume that $K < M$ and $K$ is known {\em a priori}.
	
	\section{Fusion of Algorithms for Multiple Measurement Vector Problem}
	\label{Sec:MMV_FACS_Problem}
	In this paper, we propose to employ multiple sparse reconstructions algorithms independently for estimating $\X$ from \eqref{Eqn:mmvmatrix} and fuse the resultant estimates to  yield a better sparse signal estimate. 
	Let $P \geq 2$ denote the number of different participating algorithms employed to estimate the sparse signal. Let $\hat{\mathcal{T}}_i$ denote the support-set estimated by the $i^{\text{th}}$ participating algorithm and let $\mathcal{T}$ denote the true-support-set. Denote the union of the estimated support-sets as $\Gamma$, i.e., $\Gamma \triangleq \cup_{i=1}^P \hat{\mathcal{T}}_i$, assume that $R \triangleq |\Gamma| \leq M$. We hope that different participating algorithms work on different principles and the support-set estimated by each participating algorithm includes a partially correct information about the true support-set $\T$. It may be also observed that the union of the estimated support-sets, $\Gamma$, is richer in terms of the true atoms as compared to the support-set estimated by any participating algorithm. Also note that, once the support-set is estimated, the non-zero magnitudes of $\X$ can be estimated by solving a Least-Squares (LS) problem on an over-determined system of linear equations. Hence if we can identify all the true atoms included in the joint support-set $\Gamma$, we can achieve a better sparse signal estimate. \par
Since we are estimating the support atoms only from $\Gamma$, we need to only solve the following problem which is lower dimensional as compared to the original problem \eqref{Eqn:mmvmatrix}:
\begin{align}
	\label{eqn:MMV_ApproxCS}
	\mathbf{B} = \mathbf{A}_{\Gamma} \mathbf{X}_{\Gamma,:} + \mathbf{\tilde W}, 
\end{align}
where $\A_{\Gamma}$ denotes the sub-matrix formed by the columns of $\A$ whose indices are listed in $\Gamma$, $\X_{\Gamma,:}$ denotes the submatrix formed by
the rows of $\X$ whose indices are listed in $\Gamma$, and $\mathbf{\tilde W} = \mathbf{W} + \A_{\Gamma^c} \X_{\Gamma^c, :}$. The matrix equation  \eqref{eqn:MMV_ApproxCS} represents a system of $L$ linear equations which are over-determined in nature. We use the method of {LS} to find an approximate solution to the overdetermined system of equations in \eqref{eqn:MMV_ApproxCS}. Let $\mathbf{V}_\GammaColon$ denote the {LS} solution of \eqref{eqn:MMV_ApproxCS}. We choose the support-set estimate of {MMV-FACS} as the support of $\mathbf{V}^K$, i.e., indices of those rows having the largest $\ell _2$-norm. Once the non-zero rows are identified, solving the resultant overdetermined solution using {LS} we can estimate the non-zero entries of $\hat{\mathbf{X}}$. {MMV-FACS} is summarized in Algorithm~\ref{Alg:MMV-FACS}.

\begin{algorithm}
	\caption[MMV Fusion of Algorithms for Compressed Sensing (FACS)]{: \hspace{1mm} {MMV-FACS}}
	\label{Alg:MMV-FACS}
	\begin{flushleft}
		{\bf Inputs:} $\mathbf{A} \in \R^{M \times N}$, $\mathbf{B} \in \R^{M \times
			L}$, $K$, and $\left\{ \hat\T_i \right\}_{i=1:P}	$. \\ 
		{\bf Assumption:} $|\displaystyle \mathop{\cup}_{i=1}^P$  $\THat_i| \le M$.  \\
		{\bf Initialization:} $\V = \0 \in \R^{N \times 1}$.
	\end{flushleft}
	{\bf Fusion:}
	\begin{algorithmic}[1]
		\State $\Gamma$ =         $\displaystyle \mathop{\cup}_{i=1}^P$  $\THat_i$; \hspace{1cm}
		\label{AlgStep:Gamma_MMV}
		\State $\V_\GammaColon$ = $\A_{\Gamma}^ \dag \B$,  $\V_\GammaCompColon = \0$;
		\label{AlgStep:v_Gamma_MMV}
		\State $\hat\T$ = $\mathtt{supp}(\V^K)$;  
		\label{AlgStep:T_MMV}
	\end{algorithmic}
	{\bf Outputs:} $\hat{\mathcal{T}}$ and $\hat{\mathbf{X}}$ (where $\hat{\mathbf{X}}_\THatColon = \mathbf{A}_\THatColon^\dagger \mathbf{B}$ and  $\hat{\mathbf{X}}_\THatCompColon=\mathbf{0})$
\end{algorithm}

\subsection*{Remark:} 
An alternate approach for solving an {MMV} problem is to stack all the columns of $\B$ to get a single measurement vector. Then \eqref{Eqn:mmvmatrix} in a noiseless case  becomes
\begin{align*}
\begin{bmatrix}\b_1\\ \b_2 \\ \vdots\\ \b_L\end{bmatrix}_{ML\times 1} = \left[
\begin{array}{ccccc}
\A\\
& \A & & \text{\huge0}\\
& & \ddots\\
& \text{\huge0} & & \A\\
& & & & \A
\end{array}
\right]_{ML\times NL}  \begin{bmatrix}\x_1\\ \x_2 \\ \vdots\\ \x_L\end{bmatrix}_{NL\times 1},
\end{align*} 
where $\b_i$  and $\x_i$ ($i=1,2,\dots, L$) denote the $i^\text{th}$ column of $\B$ and $\X$ respectively. 
Now, we have the following SMV problem.
\begin{align}
\begin{bmatrix}b_1\\ b_2 \\ \vdots\\ b_{NL}\end{bmatrix}_{ML\times 1} = \left[
\begin{array}{ccccc}
\A\\
& \A & & \text{\huge0}\\
& & \ddots\\
& \text{\huge0} & & \A\\
& & & & \A
\end{array}
\right]_{ML\times NL}  \begin{bmatrix} x_1\\ x_2 \\ \vdots\\ x_{NL}\end{bmatrix}_{NL\times 1}
\label{eqn:SMV}
\end{align}
In principle, we can solve \eqref{eqn:SMV} using {FACS} with sparsity level $LK$. Note that, after stacking $\X$ column-wise, we lost the joint sparsity constraint imposed on $\X$ in the MMV problem in \eqref{Eqn:mmvmatrix}. The $LK$ non-zero elements estimated from \eqref{eqn:SMV} using {FACS} can be from  more than $K$ different rows of $\X$. In the worst case, the estimate of FACS may include non-zero elements from $\min(LK, M)$ different rows of $\X$. Then we will end up with an estimate of $\X$ with $LK$ non-zero rows, which is highly undesirable. Hence stacking the columns of the observation matrix $\B$ and solving it using {FACS} is not advisable. Note that Step~3 in Algorithm~\ref{Alg:MMV-FACS} ensures that MMV-FACS estimates only $K$ non-zero rows of $\X$.

\section[Theoretical Studies of {MMV-FACS}]{Theoretical Studies of {MMV-FACS}} 
In this section, we will theoretically analyse the performance of {MMV-FACS}. We consider the general case for an arbitrary signal matrix. We also study the average case performance of {MMV-FACS} subsequently. 
\par
The performance analysis is characterized by {SRER} extended for {MMV} which is defined as 
\begin{equation}
\text{{SRER}}\triangleq \frac{\left \|\mathbf{X}\right \|_F^2}{\left \|\mathbf{X}-\hat{\mathbf{X}}\right \|_F^2}, 
\end{equation}
where $\mathbf{X}$ and $\hat{\mathbf{X}}$ denote the actual and reconstructed signal matrix respectively. 
\begin{lem}
	\label{lem1_MMVFACS}
	Suppose that A satisfies the relation, for some constant $\deltaRK \in (0,1)$,
	\begin{align*}
	&\left \|\mathbf{AX}\right \|_F \leq \sqrt{1+\delta _{R+K}} \left \|\mathbf{X}\right \|_F, 
	\end{align*}
	where $\left \|\mathbf{X}\right \|_0 \leq R+K$ and $\deltaRK \in (0,1)$. Here $\left \|\mathbf{X}\right \|_0$ denotes the number of non-zero rows of the matrix $\mathbf{X}$. Then, for every matrix $\mathbf{X}$,
	\begin{equation*}
	\left \|\mathbf{AX}\right \|_F \leq \sqrt{1+\delta _{R+K}} \Biggl[ \left \|\mathbf{X}\right \|_F + \frac{1}{\sqrt{R+K}} \left \|\mathbf{X}\right \|_{2,1}\Biggr]
	\end{equation*}
\end{lem}

\noindent \textit{Proof}: Proof is given in \appendixname~\ref{app:appendix1}.\\

\begin{lem} 
	\label{lem2_MMVFACS}
	Consider $\mathbf{A} \in \mathbb{R}^{M\times N}$ and let $\mathcal{T}_1$ \& $\mathcal{T}_2$ be two subsets of $\{1,2, \ldots N\}$ such that $\mathcal{T}_1 \cap \mathcal{T}_2 = \emptyset $. Assume that $\delta _{|\mathcal{T}_1|+|\mathcal{T}_2|} \leq 1$ and let $\mathbf{Y}$ be any matrix, such that $span(\mathbf{Y}) \in span(\mathbf{A}_{\mathcal{T}_1})$ and $\mathbf{R}=\mathbf{Y}-\mathbf{A}_{\mathcal{T}_2}\mathbf{A}_{\mathcal{T}_2}^\dagger \mathbf{Y}$. Then we have 
	\begin{equation*}
	\left(1-\frac{\delta _{|\mathcal{T}_1|+|\mathcal{T}_2|}}{1-\delta _{|\mathcal{T}_1|+|\mathcal{T}_2|}} \right) \left \|\mathbf{Y}\right \|_2 \leq \left \|\mathbf{R}\right \|_2 \leq \left \|\mathbf{Y}\right \|_2.
	\end{equation*}
\end{lem}

\noindent \textit{Proof}: Proof is given in Appendix~\ref{app:appendix2}.\\
\subsection[Performance Analysis]{Performance Analysis for Arbitrary Signals under Measurement Perturbations}
We analyse the performance of {MMV-FACS} for arbitrary signals and give an upper bound on the reconstruction error in Theorem~\ref{thm:MainMMVFACS}. We also derive a sufficient condition to get an improved performance of {MMV-FACS} scheme over any given participating algorithm. 

\begin{thm}
	\label{thm:MainMMVFACS}
	Let $\mathbf{X}$ be an arbitrary signal with $\T = \texttt{supp}(\mathbf{X}^{K})$. Consider the {MMV-FACS} setup discussed in Section~\ref{Sec:MMV_FACS_Problem}, and assume that the measurement matrix \textbf{A} satisfies {RIP} with {RIC} $\delta _{R+K}$. We have the following results:
	\begin{enumerate}
		\item Upper bound on reconstruction error: We have, 
		\begin{align*}
		\left \|\mathbf{X}-\hat{\mathbf{X}} \right \|_F \leq C_1 \left \|\mathbf{X} -\mathbf{X}^{K} \right \|_F + C_2 \left \|\mathbf{X} -\mathbf{X}^{K} \right \|_{2,1} + C_3 \left \|\mathbf{X}_\GammaCompColon\right \|_F  + \nu \left \|\mathbf{W}\right \|_F
		\end{align*}
		where $C_1=\left(1+\nu \sqrt{1+\delta _{R+K}}\right)$, $C_2= \dfrac{\nu \sqrt{1+\delta _{R+K}}}{\sqrt{R+K}}$, $C_3= \dfrac{1+\delta _{R+K}}{(1-\delta _{R+K})^2}$, and $\nu = \dfrac{3-\delta _{R+K}}{(1-\delta _{R+K})^2}$.
		\item {SRER} gain: \\
		For $\left \|\mathbf{X}_\THatiCompColon \right \|_F \neq 0$ and $\left \|\mathbf{X}_\GammaCompColon\right \|_F \neq 0$, 
		{MMV-FACS} provides at least {SRER} gain of \\$\left( \dfrac{(1-\delta_{R+K})^2}{(1+\delta_{R+K}+3\zeta+3\xi) \eta_i} \right)^2$ over the $i^{\text{th}}$ participating algorithm if \\ $\eta_i <  \dfrac{(1-\delta_{R+K})^2}{(1+\delta_{R+K}+3\zeta+3\xi)}$, where 
		$\eta _i =\dfrac{\left \|\mathbf{X}_\GammaCompColon\right \|_F}{\left \|\mathbf{X}_\THatiCompColon\right \|_F}$, $\zeta = \dfrac{\left \|\mathbf{W}\right \|_F}{\left \|\mathbf{X}_\GammaCompColon\right \|_F}$, and \\ $\xi = \left(3 \sqrt{1+\delta _{R+K}}+1\right) \dfrac{\left\|\mathbf{X}-\mathbf{X}^K\right\|_F}{3 \left\|\mathbf{X}_\GammaCompColon\right\|_F}+ \dfrac{\sqrt{1+\delta _{R+K}}}{\sqrt{R+K}}\dfrac{\left\|\mathbf{X}-\mathbf{X}^K\right\|_{2,1}}{\left\|\mathbf{X}_\GammaCompColon\right\|_F}$. 
		
	\end{enumerate}
\end{thm}

\noindent \textit{Proof:} \\
i) We have, 
\begin{align}
\left \|\mathbf{X}-\hat{\mathbf{X}}\right \|_F \leq \left \|\mathbf{X}-\mathbf{X}^{K}\right \|_F  +\left \|\mathbf{X}^{K}-\hat{\mathbf{X}}\right \|_F \label{19}
\end{align}
Consider,
\begin{align}
\left \|\mathbf{X}^K-\hat{\mathbf{X}}\right \|_F & \leq \left \|(\mathbf{X}^K)_\THatColon-\hat{\mathbf{X}}_\THatColon\right \|_F + \left \|(\mathbf{X}^K)_\THatCompColon-\hat{\mathbf{X}}_\THatCompColon\right \|_F \nonumber \\
& \leq \left \|(\mathbf{X}^K)_\THatColon-\hat{\mathbf{X}}_\THatColon\right \|_F + \left \|(\mathbf{X}^K)_\THatCompColon\right \|_F \quad (\because {\hat{\mathbf{X}}}_\THatCompColon =0) \label{five}
\end{align}

\noindent Using the relations $\hat{\mathbf{X}}_\THatColon=\mathbf{A}_{\hat{\mathcal{T}}}^\dagger \mathbf{B}$ (from Algorithm~\ref{Alg:MMV-FACS}) and $\mathbf{A}_{\hat{\mathcal{T}}}^\dagger\mathbf{A}_{\hat{\mathcal{T}}}=\mathbf{I}$, we get
\begin{align}
&\left \|(\mathbf{X}^K)_\THatColon-\hat{\mathbf{X}}_\THatColon\right \|_F \nonumber \\
\hspace{1cm}& = \left \|(\mathbf{X}^K)_\THatColon-\mathbf{A}_{\hat{\mathcal{T}}}^\dagger \mathbf{B}\right \|_F \nonumber \\
& = \left \|(\mathbf{X}^K)_\THatColon-\mathbf{A}_{\hat{\mathcal{T}}}^\dagger \left( \mathbf{AX}+\mathbf{W}\right)\right  \|_F \qquad \qquad \left(\because \mathbf{B}=\mathbf{AX}+\mathbf{W}\right)\nonumber \\ 
& = \left \|(\mathbf{X}^K)_\THatColon-\mathbf{A}_{\hat{\mathcal{T}}}^\dagger \left(\mathbf{AX}^K+\mathbf{A}(\mathbf{X}-\mathbf{X}^K)+\mathbf{W}\right)\right \|_F \nonumber \\ 
& =  \left \|(\mathbf{X}^K)_\THatColon-\mathbf{A}_{\hat{\mathcal{T}}}^\dagger \left[ \mathbf{A}_{\hat{\mathcal{T}}}(\mathbf{X}^K)_\THatColon+\mathbf{A}_{\hat{\mathcal{T}}^c}(\mathbf{X}^K)_\THatCompColon +\mathbf{A}(\mathbf{X}-\mathbf{X}^K)+\mathbf{W}\right] \right \|_F \nonumber \\
& = \left \|\mathbf{A}_{\hat{\mathcal{T}}}^\dagger \mathbf{A}_{\hat{\mathcal{T}}^c}(\mathbf{X}^K)_\THatCompColon + \mathbf{A}_{\hat{\mathcal{T}}}^\dagger \mathbf{A}(\mathbf{X} - \mathbf{X}^K) + \mathbf{A}_{\hat{\mathcal{T}}}^\dagger \mathbf{W}\right \|_F \nonumber \\
& \leq \left \|\left(\mathbf{A}_{\hat{\mathcal{T}}}^H \mathbf{A}_{\hat{\mathcal{T}}} \right)^{-1} \mathbf{A}_{\hat{\mathcal{T}}}^H\mathbf{A}_{\hat{\mathcal{T}}^c}(\mathbf{X}^K)_\THatCompColon\right \|_F +\left \|\mathbf{A}_{\hat{\mathcal{T}}}^\dagger \mathbf{A}(\mathbf{X}-\mathbf{X}^K)\right\|_F +\left\|\mathbf{A}_{\hat{\mathcal{T}}}^\dagger \mathbf{W}\right \|_F \label{fivea}
\end{align}

\noindent Let $\mathbf{x}^{(i)}$ denote the $i^\text{th}$ column of matrix $\mathbf{X}$ and $\mathbf{w}^{(i)}$ denote the $i^\text{th}$ column of matrix $\mathbf{W}$, $i=1,2,\ldots L$. Now from Proposition~3.1 and  Corollary~3.3 of \cite{Needell2009CoSaMP} we obtain the following relations.
\begin{align}
\left \|\mathbf{A}_{\hat{\mathcal{T}}}^\dagger \mathbf{w}^{(i)}\right \|_2 & \leq  \frac{ \left \|\mathbf{w}^{(i)}\right \|_2}{\sqrt{1-\delta _{R+K}}} \label{p3.1w}
\end{align}
\begin{align}
\left \|\mathbf{A}_{\hat{\mathcal{T}}}^\dagger \mathbf{A}\left(\mathbf{x}^{(i)}-(\mathbf{x}^{(i)})^{K}\right)\right \|_2 & \leq \frac{ \left \| \mathbf{A}\left(\mathbf{x}^{(i)}-(\mathbf{x}^{(i)})^{K}\right)\right \|_2}{\sqrt{1-\delta _{R+K}}} \label{p3.1a}
\end{align}
\begin{align}
\left \|\left(\mathbf{A}_{\hat{\mathcal{T}}}^H \mathbf{A}_{\hat{\mathcal{T}}} \right)^{-1} \mathbf{A}_{\hat{\mathcal{T}}}^H\mathbf{A}_{\hat{\mathcal{T}}^c}\left((\mathbf{x}^{(i)})^K\right)_{\hat{\mathcal{T}}^c}\right \|_2 \leq \frac{\delta _{R+K}}{1-\delta _{R+K}}\left \|\left((\mathbf{x}^{(i)})^{K}\right)_{\hat{\mathcal{T}}^c}\right \|_2 \label{p3.1b}
\end{align}

Consider \eqref{p3.1w}, we get
\begin{align}
\left \|\mathbf{A}_{\hat{\mathcal{T}}}^\dagger \mathbf{w}^{(i)}\right \|_2^2 \leq  \frac{ \left \|\mathbf{w}^{(i)}\right \|_2^2}{1-\delta _{R+K}} \quad \forall \; i=1,2,\ldots L \nonumber 
\end{align}	
Summing the above equation over $i=1,2,\ldots L$, we obtain
\begin{align}	
\sum _{i=1}^L\left \|\mathbf{A}_{\hat{\mathcal{T}}}^\dagger \mathbf{w}^{(i)}\right \|_2^2 & \leq \frac{1}{1-\delta _{R+K}}\sum _{i=1}^L \left \|\mathbf{w}^{(i)}\right \|_2^2 \nonumber \\
\left \|\mathbf{A}_{\hat{\mathcal{T}}}^\dagger \mathbf{W}\right \|_F^2 & \leq \frac{1}{1-\delta _{R+K}}\left \|\mathbf{W}\right \|_F^2 \nonumber \\
\left \|\mathbf{A}_{\hat{\mathcal{T}}}^\dagger \mathbf{W}\right \|_F & \leq \frac{1}{\sqrt{1-\delta _{R+K}}}\left \|\mathbf{W}\right \|_F. \label{p3.1wa}
\end{align}	
Similarly, summing the relations in \eqref{p3.1a} and \eqref{p3.1b}, we obtain 
\begin{align}
\left \|\mathbf{A}_{\hat{\mathcal{T}}}^\dagger \mathbf{A}\left(\mathbf{X}-\mathbf{X}^K\right)\right \|_F & \leq  \frac{ \left \| \mathbf{A}\left(\mathbf{X}-\mathbf{X}^K\right)\right \|_F}{\sqrt{1-\delta _{R+K}}} \label{p3.1aa}
\end{align}
\begin{align}
\left \|\left(\mathbf{A}_{\hat{\mathcal{T}}}^H \mathbf{A}_{\hat{\mathcal{T}}} \right)^{-1} \mathbf{A}_{\hat{\mathcal{T}}}^H\mathbf{A}_{\hat{\mathcal{T}}^c}(\mathbf{X}^K)_\THatCompColon\right \|_F \leq \frac{\delta _{R+K}}{1-\delta _{R+K}}\left \|(\mathbf{X}^K)_\THatCompColon\right \|_F \label{p3.1bb}
\end{align}
Substituting \eqref{p3.1wa},\eqref{p3.1aa} and \eqref{p3.1bb} in \eqref{fivea}, we get
\begin{align}
\left \|(\mathbf{X}^K)_\THatColon-\hat{\mathbf{X}}_\THatColon\right \|_F & \leq \frac{\delta _{R+K}}{1-\delta _{R+K}}\left \|(\mathbf{X}^K)_\THatCompColon\right \|_F +\frac{ \left \| \mathbf{A}\left(\mathbf{X}-\mathbf{X}^K\right)\right \|_F}{\sqrt{1-\delta _{R+K}}} + \frac{ \left \|\mathbf{W}\right \|_F}{\sqrt{1-\delta _{R+K}}} \nonumber \\
& \leq \frac{\delta _{R+K}}{1-\delta _{R+K}}\left \|(\mathbf{X}^K)_\THatCompColon\right \|_F +\frac{1}{1-\delta _{R+K}} \left( \left \| \mathbf{A}\left(\mathbf{X}-\mathbf{X}^K\right)\right \|_F+\left \|\mathbf{W}\right \|_F \right) \label{twenty}
\end{align}	
Substituting \eqref{twenty} in \eqref{five}, we get 
\begin{align}
\left \|\mathbf{X}^{K} -\hat{\mathbf{X}}\right \|_F & \leq \frac{1}{1-\delta _{R+K}} \left \|(\mathbf{X}^K)_\THatCompColon\right \|_F + \frac{1}{1-\delta _{R+K}} \left( \left \|\mathbf{A}(\mathbf{X}-\mathbf{X}^{K}) \right \|_F + \left \|\mathbf{W}\right \|_F \right) \label{20}
\end{align}	
Next, we will find an upper bound for $\left \|(\mathbf{X}^K)_\THatCompColon\right \|_F$. \\
Define $\hat{\mathcal{T}}_\Delta \triangleq \Gamma \setminus \hat{\mathcal{T}}$.  That is, $\hat{\mathcal{T}}_\Delta$ is the set formed by the atoms in $\Gamma$ which are discarded by Algorithm~\ref{Alg:MMV-FACS}. Since $\hat{\mathcal{T}} \subset \Gamma$, we have $\hat{\mathcal{T}}^c = \Gamma ^c \cup \hat{\mathcal{T}}_\Delta$ and hence we obtain
\begin{align}
\left \|(\mathbf{X}^K)_\THatCompColon\right \|_F \leq \left \|(\mathbf{X}^K)_\GammaCompColon\right \|_F+\left \|(\mathbf{X}^K)_\THatDeltaColon\right \|_F \label{eight}
\end{align}	
We also have, 
\begin{align}
\left \|(\mathbf{X}^K)_\THatDeltaColon\right \|_F & \leq \left \|\left(\mathbf{V}_\GammaColon \right)_\THatDeltaColon\right \|_F +\left \|\left(\mathbf{V}_\GammaColon -(\mathbf{X}^K)_\GammaColon \right)_\THatDeltaColon\right \|_F \nonumber \\
& \leq \left \|\left(\mathbf{V}_\GammaColon \right)_\THatDeltaColon\right \|_F +\left \|\mathbf{V}_\GammaColon -(\mathbf{X}^K)_\GammaColon\right \|_F \label{nine}
\end{align}	
Note that $\left(\mathbf{V}_\GammaColon \right)_\THatColon$ contains the $K$-rows of $\mathbf{V}_\GammaColon$ with highest row $\ell_2$-norm. Therefore, using $|\hat{\mathcal{T}}|=|\mathcal{T}|=K$, we get
\begin{align}
\left \|\left(\mathbf{V}_\GammaColon \right)_\THatDeltaColon\right \|_F & \leq \left \|\left(\mathbf{V}_\GammaColon \right)_{\Gamma \setminus \mathcal{T}, \, :}\right \|_F \nonumber \\
& = \left \|\mathbf{V}_{\Gamma \setminus \mathcal{T}, \, :} -(\mathbf{X}^K)_{\Gamma \setminus \mathcal{T}, \, :}\right \|_F \quad \left(\because (\mathbf{X}^K)_{\Gamma \setminus \mathcal{T}, \, :} =0 \right)\nonumber \\
& \leq \left \|\left(\mathbf{V}_\GammaColon -(\mathbf{X}^K)_\GammaColon\right)\right \|_F. \label{ten}
\end{align}	
Substituting \eqref{ten} in \eqref{nine}, we get 
\begin{align}
\left \|(\mathbf{X}^K)_\THatDeltaColon\right \|_F \leq 2\left \|\left(\mathbf{V}_\GammaColon - (\mathbf{X}^K)_\GammaColon\right)\right \|_F. \label{eleven}
\end{align}
Now, consider
\begin{flalign}
&\left \|\mathbf{V}_\GammaColon -(\mathbf{X}^K)_\GammaColon\right \|_F \hspace{4cm} \nonumber \\
&\hspace{1cm} = \left \|\mathbf{A}_\Gamma^{\dagger} \mathbf{B} -(\mathbf{X}^K)_\GammaColon\right \|_F \nonumber \\
&\hspace{1cm} = \left \|\mathbf{A}_\Gamma^{\dagger} (\mathbf{AX}+\mathbf{W}) -(\mathbf{X}^K)_\GammaColon\right \|_F \nonumber \\
&\hspace{1cm} = \left \|\mathbf{A}_\Gamma^{\dagger} \left( \mathbf{AX}^K+\mathbf{A}(\mathbf{X}-\mathbf{X}^K)+\mathbf{W} \right) - (\mathbf{X}^K)_\GammaColon\right \|_F \nonumber \\
&\hspace{1cm} =  \left \|\mathbf{A}_\Gamma^{\dagger} \Bigl(\mathbf{A}_\Gamma (\mathbf{X}^K)_\GammaColon +\mathbf{A}_{\Gamma ^c} (\mathbf{X}^K)_\GammaCompColon + \mathbf{A}(\mathbf{X}-\mathbf{X}^K)+\mathbf{W}\Bigr) -(\mathbf{X}^K)_\GammaColon \right\|_F \nonumber \\
&\hspace{1cm} = \left \|\mathbf{A}_\Gamma^{\dagger} \mathbf{A}_{\Gamma ^c} (\mathbf{X}^K)_\GammaCompColon + \mathbf{A}_\Gamma^{\dagger} \mathbf{A}(\mathbf{X}-\mathbf{X}^K)+\mathbf{A}_\Gamma^{\dagger} \mathbf{W} \right \|_F \qquad (\because \A^\dagger_\Gamma\A_\Gamma = \I)\nonumber  \\
&\hspace{1cm} \leq \left \|\mathbf{A}_\Gamma^{\dagger} \mathbf{A}_{\Gamma ^c} (\mathbf{X}^K)_\GammaCompColon \right \|_F + \left \|\mathbf{A}_\Gamma^{\dagger} \mathbf{A}(\mathbf{X}-\mathbf{X}^K)\right \|_F +\ \left \| \mathbf{A}_\Gamma^{\dagger} \mathbf{W} \right \|_F. \label{eqn:MMV_Soo1}
\end{flalign}	

Using \eqref{p3.1wa}, \eqref{p3.1aa} and \eqref{p3.1bb} in \eqref{eqn:MMV_Soo1}, we get
\begin{align}
\left \|\mathbf{V}_\GammaColon -(\mathbf{X}^K)_\GammaColon\right \|_F &\leq  \frac{\delta _{R+K}}{1-\delta _{R+K}}\left \|(\mathbf{X}^K)_\GammaCompColon\right \|_F +\frac{\left \|\mathbf{A}(\mathbf{X}-\mathbf{X}^K)\right \|_F}{\sqrt{1-\delta _{R+K}}}  +\frac{ \left \|\mathbf{W}\right \|_F}{{\sqrt{1-\delta _{R+K}}}} \nonumber \\	
& \leq  \frac{\delta _{R+K}}{1-\delta _{R+K}}\left \|(\mathbf{X}^K)_\GammaCompColon\right \|_F +\frac{\left \|\mathbf{A}(\mathbf{X}-\mathbf{X}^K)\right \|_F}{{1-\delta _{R+K}}}  +\frac{ \left \|\mathbf{W}\right \|_F}{{{1-\delta _{R+K}}}}.
\label{twelve}\\
& \qquad (\because 0< 1-\deltaRK < 1) \nonumber 
\end{align}	

Using \eqref{eleven} and \eqref{twelve} in \eqref{eight}, we get 
\begin{align}
\left \|(\mathbf{X}^K)_\THatCompColon\right \|_F \leq & \frac{1+\delta _{R+K}}{1-\delta _{R+K}}\left \|(\mathbf{X}^K)_\GammaCompColon\right \|_F +\frac{2}{{1-\delta _{R+K}}}\left( \left \|\mathbf{A}(\mathbf{X}-\mathbf{X}^{K}) \right \|_F + \left \|\mathbf{W}\right \|_F \right). \label{13} 
\end{align}	
Let $\mathbf{x}_1^{(i)}$ denote the $i^\text{th}$ column of matrix $\mathbf{X}^{K}$. The, we have,
\begin{align}
\left \|\mathbf{A}\mathbf{x}_1^{(i)}\right \|_2^2 & \leq (1+\delta _{R+K})\left \|\mathbf{x}_1^{(i)}\right \|_2^2 &&(\because \A \text{ satisfies {RIP}} )\nonumber \\
\sum _{i=1}^L \left \|\mathbf{A}\mathbf{x}_1^{(i)}\right \|_2^2 & \leq \sum _{i=1}^L(1+\delta _{R+K})\left \|\mathbf{x}_1^{(i)}\right \|_2^2 \nonumber \\
\left \|\mathbf{A}\mathbf{X}^{K}\right \|_F^2 & \leq (1+\delta _{R+K})\left \|\mathbf{X}^{K}\right \|_F^2 \label{21a}
\end{align}	
\noindent Using Lemma~\ref{lem1_MMVFACS} and \eqref{21a}, we get
\begin{align}
\left \|\mathbf{A}(\mathbf{X}-\mathbf{X}^{K})\right \|_F \leq \sqrt{1+\delta _{R+K}} \Biggl[ \left \|\mathbf{X}-\mathbf{X}^{K}\right \|_F + \frac{1}{\sqrt{R+K}} \left \|\left(\mathbf{X}-\mathbf{X}^{K}\right)\right \|_{2,1}\Biggr] \label{22}
\end{align}

Substituting \eqref{13} in \eqref{20}, we get 
\begin{align}	
\vecNorm{\X^K - \XHat}_F	 &\le \frac{1+\deltaRK}{(1-\deltaRK)^2} \vecNorm{(\X^K)_\GammaCompColon}_F + \frac{3-\deltaRK}{(1-\deltaRK)^2} \left[\vecNorm{\A(\X-\X^K)}_F + \vecNorm{\W}_F\right] \nonumber \\
&\le  \nu \sqrt{1+\delta _{R+K}} \left \|\mathbf{X}-\mathbf{X}^{K}\right \|_F + \frac{\nu \sqrt{1+\delta _{R+K}}\left \|\mathbf{X}-\mathbf{X}^{K}\right \|_{2,1}}{\sqrt{R+K}} \nonumber \\ 
& \quad + \frac{1+\delta _{R+K}}{(1-\delta _{R+K})^2}\left \|\mathbf{X}_\GammaCompColon\right \|_F + \nu \left \|\mathbf{W}\right \|_F, \qquad \qquad \text{(using }\eqref{22}\text{)}\label{23}
\end{align}	
where $\displaystyle \nu=\frac{3-\deltaRK}{(1-\delta_{R+K})^2}$.\\
Substituting \eqref{23} in \eqref{19} and using the definitions of $C_1, C_2$, and $C_3$, we get 
\begin{align}
\left \|\mathbf{X}-\hat{\mathbf{X}}\right \|_F \leq \; & C_1 \left \|\mathbf{X}-\mathbf{X}^{K}\right \|_F + C_2 \left \|\mathbf{X}-\mathbf{X}^{K}\right \|_{2,1} + C_3 \left \|\mathbf{X}_\GammaCompColon\right \|_F + \nu \left \|\mathbf{W}\right \|_F. \label{result}
\end{align}
\\	
ii) Using \eqref{result} and the definitions of $\xi$ and $\eta _i$, we get 
\begin{align}
\left \| \mathbf{X}- \hat{\mathbf{X}}\right \|_F & \leq \frac{1+\delta _{R+K}+3\zeta+3\xi}{(1-\delta _{R+K})^2}\left \|\mathbf{X}_\GammaCompColon\right \|_F \nonumber \\
& = \frac{1+\delta _{R+K}+3\zeta+3\xi}{(1-\delta _{R+K})^2} \eta _i \left \|(\mathbf{X}-{\hat{\mathbf{X}}}_i)_{\hat{\mathcal{T}}_i^c, \, :}\right \|_F \quad(\because ({\hat{\mathbf{X}}}_i)_{\hat{\mathcal{T}}_i^c, \, :} = \mathbf{0}) \nonumber \\
& \leq \frac{1+\delta _{R+K}+3\zeta+3\xi}{(1-\delta _{R+K})^2} \eta _i \left \|(\mathbf{X}-{\hat{\mathbf{X}}}_i)\right \|_F. \nonumber 
\end{align}	
Hence, we obtain the relation for {SRER} for {MMV-FACS}, in case of arbitrary signals, as
\begin{align*}
\text{{SRER}}|_{\text{{MMV-FACS}}} & = \frac{\left \|\mathbf{X}\right \|_F^2}{\left \|\mathbf{X}-\hat{\mathbf{X}}\right \|_F^2} \\
& \geq \frac{\left \|\mathbf{X}\right \|^2_F}{\left \|\mathbf{X}-{\hat{\mathbf{X}}}_i\right \|^2_F} \times \left( \frac{(1-\delta_{R+K})^2}{(1+\delta_{R+K}+3\zeta+3 \xi) \eta_i} \right)^2
\end{align*}	
Hence {MMV-FACS} provides at least {SRER} gain of $\left( \dfrac{(1-\delta_{R+K})^2}{(1+\delta_{R+K}+3\zeta+3\xi) \eta_i} \right)^2$ over $i^\text{th}$ algorithm if $\eta_i <  \dfrac{(1-\delta_{R+K})^2}{(1+\delta_{R+K}+3\zeta+3\xi)}$. 	\\
Note that $\dfrac{(1-\delta_{R+K})^2}{(1+\delta_{R+K}+3\zeta+3\xi)} < 1$.
\qed

\subsection{Exactly $K$-sparse Matrix}
Theorem~\ref{thm:MainMMVFACS} considered the case when $\mathbf{X}$ is an arbitrary matrix. If $\mathbf{X}$ is a $K$-sparse matrix then we have $\mathbf{X}=\mathbf{X}^K$ and $\xi=0$. Thus, it follows from Theorem~\ref{thm:MainMMVFACS} that, {MMV-FACS} provides at least {SRER} gain of $\left( \dfrac{(1-\delta_{R+K})^2}{(1+\delta_{R+K}+3\zeta) \eta_i} \right)^2$ over $i^\text{th}$ participating algorithm if $\eta_i <  \dfrac{(1-\delta_{R+K})^2}{(1+\delta_{R+K}+3\zeta)} $. Thus, the improvement in the {SRER} gain provided by {MMV-FACS} over the $i^\text{th}$ Algorithm for a $K$-sparse matrix is greater than that of an arbitrary matrix by a factor of $\left(1+\dfrac{3 \xi}{(1+\delta_{R+K}+3\zeta)}\right)^2$.\\
\par
The second part of Theorem~\ref{thm:MainMMVFACS} considers the case when $\left \|\mathbf{X}_\THatiCompColon \right \|_F \neq 0 $ and $\left \|\mathbf{X}_\GammaCompColon\right \|_F \neq 0$. If $\left \|\mathbf{X}_\THatiCompColon \right \|_F \neq 0 $, then $\hat{\mathcal{T}}_i \nsubseteq \mathcal{T}$. Also, $\left \|\mathbf{X}_\GammaCompColon\right \|_F = 0$ implies $ \mathcal{T} \subseteq \Gamma$. Suppose $\left \|\mathbf{X}_\THatiCompColon \right \|_F = 0 $, then the support-set is correctly estimated by $i^\text{th}$ algorithm and further performance improvement is not possible by {MMV-FACS}. 
Hence we consider the case where $\left \|\mathbf{X}_\GammaCompColon\right \|_F = 0$, and derive the condition for exact reconstruction by {MMV-FACS} in the following proposition.
\begin{prop}
	Assume that $\left \|\mathbf{X}_\GammaCompColon\right \|_F = 0$ and all other conditions in Theorem~\ref{thm:MainMMVFACS} hold good. Then, in clean measurement case ($\mathbf{W}=\mathbf{0}$), {MMV-FACS} estimates the support-set correctly and provides exact reconstruction.
\end{prop}
\noindent \textit{Proof}: 
We have 
\begin{align}
\mathbf{X}_\GammaCompColon & = \mathbf{0} \Rightarrow \mathcal{T} \subset \Gamma \label{16}\\
\mathbf{B} & =\mathbf{A}_{\mathcal{T}}\mathbf{X}_\THatColon+\mathbf{W} \label{17}
\end{align}
From Algorithm~\ref{Alg:MMV-FACS}, we have $\mathbf{V} \in \mathbb{R}^{N\times L}$ where $\mathbf{V}_{\Gamma^c,\, :}=\mathbf{0}$, and 
\begin{align*}
\mathbf{V}_\GammaColon & = \mathbf{A}_{\Gamma}^\dagger \mathbf{B} \\
& = \mathbf{A}_{\Gamma}^\dagger (\mathbf{A}_{\mathcal{T}}\mathbf{X}_\THatColon + \mathbf{W}) &&(\text{using} \;\eqref{17})\\
& = \mathbf{A}_{\Gamma}^\dagger (\mathbf{A}_{\Gamma}\mathbf{X}_\GammaColon + \mathbf{W}) &&(\text{using} \;\eqref{16})\\
& = \mathbf{X}_\GammaColon + \mathbf{A}_{\Gamma}^\dagger \mathbf{W}.
\end{align*}
If $\mathbf{W}=\mathbf{0}$, then $\mathbf{V}_\GammaColon = \mathbf{X}_\GammaColon$ and $\mathbf{V}=\mathbf{X}$ ($\because \mathcal{T} \subset \Gamma $). Thus {MMV-FACS} estimates the support-set correctly from $\mathbf{V}$.
\qed
\par
In practice, the original signal is not known and hence it is not possible to evaluate the performance w.r.t. the true signal. Hence in applications, the decrease in energy of the residual is often treated as a measure of performance improvement. Proposition \ref{ptwo} gives a sufficient condition for decrease in the energy of the residual matrix obtained by {MMV-FACS} over the $i^\text{th}$ participating algorithm.

\begin{prop}\label{ptwo}
	For a $K$-sparse matrix $\mathbf{X}$, let $\mathbf{R}=\mathbf{B}-\mathbf{A}_{\hat{\mathcal{T}}}\mathbf{A}_{\hat{\mathcal{T}}}^\dagger \mathbf{B}$ and $\mathbf{R}_i=\mathbf{B}-\mathbf{A}_{\hat{\mathcal{T}}_i}\mathbf{A}_{\hat{\mathcal{T}}_i}^\dagger \mathbf{B}$ represent the residue matrix of {MMV-FACS} and $i^\text{th}$ Algorithm respectively. Assume that $\dfrac{\sqrt{1+\delta _{R+K}}}{1-\delta _{R+K}} ( 1 + \delta _{R+K} + 3 \zeta) \leq \left(\dfrac{1-2 \delta _{R+K}}{\eta _i \sqrt{1-\delta _{R+K}}} -\zeta \right) $ is satisfied then we have, $\left \|\mathbf{R}\right \|_F \leq \left \|\mathbf{R}_i\right \|_F$.
\end{prop}
\begin{proof} We have,
	\begin{align*}
	\left \|\mathbf{R}\right \|_F & = \left \|\mathbf{B}-\mathbf{A}_{\hat{\mathcal{T}}}\mathbf{A}_{\hat{\mathcal{T}}}^\dagger \mathbf{B}\right \|_F \\
	& = \left \|\mathbf{AX}+\mathbf{W}-\mathbf{A}_{\hat{\mathcal{T}}}\mathbf{A}_{\hat{\mathcal{T}}}^\dagger \left( \mathbf{AX}+\mathbf{W}\right)\right \|_F \quad \left(\because \mathbf{B}=\mathbf{AX}+\mathbf{W}\right)\\ 
	& = \left \|\mathbf{A}_{\hat{\mathcal{T}}}\mathbf{X}_\THatColon+\mathbf{A}_{\hat{\mathcal{T}}^c}\mathbf{X}_\THatCompColon+\mathbf{W} -\mathbf{A}_{\hat{\mathcal{T}}}\mathbf{A}_{\hat{\mathcal{T}}}^\dagger \left( \mathbf{A}_{\hat{\mathcal{T}}}\mathbf{X}_\THatColon+\mathbf{A}_{\hat{\mathcal{T}}^c}\mathbf{X}_\THatCompColon+\mathbf{W}\right)\right \|_F \\
	& = \left \|\mathbf{A}_{\hat{\mathcal{T}}^c}\mathbf{X}_\THatCompColon-\mathbf{A}_{\hat{\mathcal{T}}}\mathbf{A}_{\hat{\mathcal{T}}}^\dagger \mathbf{A}_{\hat{\mathcal{T}}^c}\mathbf{X}_\THatCompColon +\mathbf{W}-\mathbf{A}_{\hat{\mathcal{T}}}\mathbf{A}_{\hat{\mathcal{T}}}^\dagger \mathbf{W}\right \|_F  (\because \mathbf{A}_{\hat{\mathcal{T}}}^\dagger \mathbf{A}_{\hat{\mathcal{T}}}=\mathbf{I}) \\
	& \leq \left \|\mathbf{A}_{\hat{\mathcal{T}}^c}\mathbf{X}_\THatCompColon-\mathbf{A}_{\hat{\mathcal{T}}}\mathbf{A}_{\hat{\mathcal{T}}}^\dagger \mathbf{A}_{\hat{\mathcal{T}}^c}\mathbf{X}_\THatCompColon\right \|_F +\left \|\mathbf{W}-\mathbf{A}_{\hat{\mathcal{T}}}\mathbf{A}_{\hat{\mathcal{T}}}^\dagger \mathbf{W}\right \|_F \\
	& \leq \left \|\mathbf{A}_{\hat{\mathcal{T}}^c}\mathbf{X}_\THatCompColon\right \|_F + \left \|\mathbf{W}\right \|_F\quad \left(\text{Using Lemma~\ref{lem2_MMVFACS}}\right)\\
	& = \left \|\mathbf{A}_{\mathcal{T}\setminus\hat{\mathcal{T}}}\mathbf{X}_{\mathcal{T}\setminus\hat{\mathcal{T}}, \, :}\right \|_F + \left \|\mathbf{W}\right \|_F \quad \left(\because \mathcal{T}=\text{supp}(\mathbf{X})\right) \\
	& \leq \sqrt{1+\delta _{R+K}} \left \|\mathbf{X}_\THatCompColon\right \|_F + \left \|\mathbf{W}\right \|_F \quad \left(\because |\mathcal{T}\setminus\hat{\mathcal{T}}| \leq K \quad \& \quad \delta _K \leq \delta _{R+K}\right).
	\end{align*}
	
	\noindent Using \eqref{13} we have, 
	\begin{align}
	\left \|\mathbf{R}\right \|_F & \leq \sqrt{1+\delta _{R+K}} \left( \frac{1+\delta _{R+K}}{1-\delta _{R+K}}\left \|\mathbf{X}_\GammaCompColon\right \|_F + \frac{ 2}{{1-\delta _{R+K}}}  \left \|\mathbf{W}\right \|_F\right) + \left \|\mathbf{W}\right \|_F \nonumber \\
	& \leq \frac{\sqrt{1+\delta _{R+K}}}{1-\delta _{R+K}} \left( 1 + \delta _{R+K} \nonumber + 2 \zeta + \frac{1-\delta _{R+K}}{\sqrt{1+\delta _{R+K}}}\zeta\right) \left \|\mathbf{X}_\GammaCompColon\right \|_F \nonumber \\
	& \leq \frac{\sqrt{1+\delta _{R+K}}}{1-\delta _{R+K}} \left( 1 + \delta _{R+K} + 3 \zeta \right) \left \|\mathbf{X}_\GammaCompColon\right \|_F. \label{r}
	\end{align}
	
	\noindent Now, consider 
	\begin{align}
	\left \|\mathbf{R}_i\right \|_F & = \left \|\mathbf{B}-\mathbf{A}_{\hat{\mathcal{T}}_i}\mathbf{A}_{\hat{\mathcal{T}}_i}^\dagger \mathbf{B}\right \|_F \nonumber \\
	& = \left \|\mathbf{AX}+\mathbf{W}-\mathbf{A}_{\hat{\mathcal{T}}_i}\mathbf{A}_{\hat{\mathcal{T}}_i}^\dagger \left( \mathbf{AX}+\mathbf{W}\right)\right \|_F \quad (\because \mathbf{B}=\mathbf{AX}+\mathbf{W})\nonumber \\ 
	& =  \left \|\mathbf{A}_{\hat{\mathcal{T}}_i}\mathbf{X}_\THatiColon+\mathbf{A}_{\hat{\mathcal{T}}_i^c}\mathbf{X}_\THatiCompColon+\mathbf{W} -\mathbf{A}_{\hat{\mathcal{T}}_i}\mathbf{A}_{\hat{\mathcal{T}}_i}^\dagger \left( \mathbf{A}_{\hat{\mathcal{T}}_i}\mathbf{X}_\THatiColon+\mathbf{A}_{\hat{\mathcal{T}}_i^c}\mathbf{X}_\THatiCompColon+\mathbf{W}\right) \right\|_F \nonumber \\
	& = \left \|\mathbf{A}_{\hat{\mathcal{T}}_i^c}\mathbf{X}_\THatiCompColon+\mathbf{W} -\mathbf{A}_{\hat{\mathcal{T}}_i}\mathbf{A}_{\hat{\mathcal{T}}_i}^\dagger \mathbf{A}_{\hat{\mathcal{T}}_i^c}\mathbf{X}_\THatiCompColon -\mathbf{A}_{\hat{\mathcal{T}}_i}\mathbf{A}_{\hat{\mathcal{T}}_i}^\dagger \mathbf{W}\right \|_F \nonumber\\
	& \geq \left \|\mathbf{A}_{\hat{\mathcal{T}}_i^c}\mathbf{X}_\THatiCompColon-\mathbf{A}_{\hat{\mathcal{T}}_i}\mathbf{A}_{\hat{\mathcal{T}}_i}^\dagger \mathbf{A}_{\hat{\mathcal{T}}_i^c}\mathbf{X}_\THatiCompColon\right \|_F-\left \|\mathbf{W}-\mathbf{A}_{\hat{\mathcal{T}}_i}\mathbf{A}_{\hat{\mathcal{T}}_i}^\dagger \mathbf{W}\right \|_F \nonumber \\ 
	& \quad (\text{Using reverse triangle inequality}) \nonumber \\
	& = \left \|\mathbf{A}_{\mathcal{T}\setminus \hat{\mathcal{T}}_i }\mathbf{X}_{\mathcal{T}\setminus \hat{\mathcal{T}}_i,\,: } - \mathbf{A}_{\hat{\mathcal{T}}_i }\mathbf{A}_{\hat{\mathcal{T}}_i }^\dagger \mathbf{A}_{\mathcal{T}\setminus \hat{\mathcal{T}}_i }\mathbf{X}_{\mathcal{T}\setminus \hat{\mathcal{T}}_i,\,: }\right \|_F -\left \|\mathbf{W}-\mathbf{A}_{\hat{\mathcal{T}}_i}\mathbf{A}_{\hat{\mathcal{T}}_i}^\dagger \mathbf{W}\right \|_F \nonumber \\
	& \geq \left( 1- \frac{\delta _{R+K}}{1-\delta _{R+K}} \right) \left \|\mathbf{A}_{\mathcal{T}\setminus \hat{\mathcal{T}}_i }\mathbf{X}_{\mathcal{T}\setminus \hat{\mathcal{T}}_i,\,: } \right \|_F - \left \|\mathbf{W}\right \|_F \nonumber \\
	& \quad (\text{Using Lemma~\ref{lem2_MMVFACS}} \text{ \& } \delta _{2K} \leq \delta _{R+K}) \nonumber \\
	& \geq \frac{1-2 \delta_{R+K}}{1-\delta _{R+K}} \sqrt{1-\delta _{R+K}} \left \|\mathbf{X}_\THatiCompColon\right \|_F - \left \|\mathbf{W}\right \|_F (\because |\mathcal{T}\setminus \hat{\mathcal{T}}| \leq K \quad \& \quad \delta _K \leq \delta _{R+K}) \nonumber \\ 
	& = \left( \frac{1-2 \delta _{R+K}}{\eta _i \sqrt{1-\delta _{R+K}}} -\zeta \right) \left \|\mathbf{X}_\GammaCompColon\right \|_F. \label{ri}
	\end{align}
	
	\noindent From \eqref{r} and \eqref{ri} we get a sufficient condition for $\left \|\mathbf{R}\right \|_F \leq \left \|\mathbf{R}_i\right \|_F$ as 
	\begin{align}
	\label{eqn:Prop2}
	\frac{\sqrt{1+\delta _{R+K}}}{1-\delta _{R+K}} ( 1 + \delta _{R+K} + 3 \zeta) \leq \left( \frac{1-2 \delta _{R+K}}{\eta _i \sqrt{1-\delta _{R+K}}} -\zeta \right). 
	\end{align}
	
	\noindent Thus, if \eqref{eqn:Prop2} is satisfied, {MMV-FACS} produces a smaller residual matrix (in the Frobenius norm sense) than that of the $i^\text{th}$ participating algorithm.
\end{proof}
\subsection{Average Case Analysis}
Intuitively, we expect multiple measurement vector problem to perform better than the single measurement vector case. However, if each measurement vector is the same, i.e., in the worst case, we have $\mathbf{x}^{(i)}=\mathbf{c}, \; \forall \; i=1,\ldots,L$, then we do not have any additional information on $\mathbf{X}$ than that provided by a single measurement vector $\mathbf{x}^{(1)}$. So far we have carried out only the worst case analysis, i.e., conditions under which the algorithm is able to recover any joint sparse matrix \textbf{X}. This approach does not provide insight into the superiority of sparse signal reconstruction with multiple measurement vectors compared to the single measurement vector case. 

To notice a performance gain with multiple measurement vectors, next we proceed with an average case analysis. Here we impose a probability model on the $K$ sparse $\mathbf{X}$ as suggested by Remi {\em et al.}~\cite{Remi2008AtomsOfAll}. In particular, on the support-set $\mathcal{T}$, we impose that $\mathbf{X}_\TColon=\Sigma \Phi$, where $\Sigma$ is a $K \times K$ diagonal matrix with positive diagonal entries and $\Phi$ is a $ K \times L$ random matrix with {i.i.d.}\ Gaussian entries. Our goal is to show that, under this signal model, the typical behaviour of {MMV-FACS} is better than in the worst case. 
\allowdisplaybreaks
\begin{thm}
	Consider the {MMV-FACS} setup discussed in Section~\ref{Sec:MMV_FACS_Problem}. Assume a Gaussian signal model, i.e., $\mathbf{X}_\TColon=\Sigma \Phi$, where $\Sigma$ is a $K \times K$ diagonal matrix with positive diagonal entries and $\Phi$ is a $ K \times L$ random matrix with {i.i.d.}\ Gaussian entries. Let $\mathbf{e}_i$ denote a $|\Gamma| \times 1$ vector with a `1' in the $i^\text{th}$ coordinate and `0' elsewhere. Let $\eta = \displaystyle\esupmin_{i \in (\mathcal{T} \cap \Gamma)} \left \|\mathbf{e}_i ^T \mathbf{A}_{\Gamma} ^{\dagger}\mathbf{W}\right \|_2 +  \esupmax_{j \in (\Gamma \setminus \mathcal{T})} \left \|\mathbf{e}_j ^T \mathbf{A}_{\Gamma} ^{\dagger}\mathbf{W}\right \|_2 $ and 
	\begin{align*}
	\gamma = \frac{\displaystyle \esupmin_{i \in (\mathcal{T} \cap \Gamma)} \left \|\mathbf{e}_i ^T \mathbf{A}_{\Gamma} ^{\dagger} \A_\T \Sigma\right \|_2 - \esupmax_{j \in (\Gamma \setminus \mathcal{T})} \left \|\mathbf{e}_j ^T \mathbf{A}_{\Gamma} ^{\dagger} \A_\T \Sigma \right \|_2 - \dfrac{\eta}{C_2(L)}}{\displaystyle \esupmin_{i \in (\mathcal{T} \cap \Gamma)} \left \|\mathbf{e}_i ^T \mathbf{A}_{\Gamma} ^{\dagger}\A_\T \Sigma\right \|_2 +  \esupmax_{j \in (\Gamma \setminus \mathcal{T})} \left \|\mathbf{e}_j ^T \mathbf{A}_{\Gamma} ^{\dagger}\A_\T \Sigma\right \|_2}.
	\end{align*}
	where $C_2(L)=\E \left \|Z\right \|_2$ with $Z=(Z_1, \ldots, Z_L) $ being a vector of independent standard normal variables. Assume that $\displaystyle \esupmin_{i \in (\mathcal{T} \cap \Gamma)} \left \|\mathbf{e}_i ^T \mathbf{A}_{\Gamma}^{\dagger}\A_\T \Sigma\right \|_2 - \esupmax_{j \in (\Gamma \setminus \mathcal{T})} \left \|\mathbf{e}_j ^T \mathbf{A}_{\Gamma} ^{\dagger} \A_\T \Sigma\right \|_2 > \dfrac{\eta}{C_2(L)}$. Let $\Theta$ denote the event that {MMV-FACS} picks all correct indices from the union-set $\Gamma$. Then, we have, 
	\begin{align*}
	P(\Theta) \geq 1-K \exp (-2 A_2(L) \gamma ^2),
	\end{align*}
	where $\displaystyle A_2(L)= \left(\frac{\ddot{\Gamma}(\frac{L+1}{2})}{\ddot{\Gamma}(\frac{L}{2})}\right)^2 \approx \frac{L}{2}$,  $\ddot{\Gamma}(\cdot)$ denotes the Gamma function.
\end{thm}
\begin{proof} We have, 
	\begin{align}
	P(\Theta) & = P\left(\esupmin_{i \in (\mathcal{T} \cap \Gamma)} \left \|\mathbf{e}_i ^T \mathbf{A}_{\Gamma} ^{\dagger} \mathbf{B}\right \|_2 > \esupmax_{j \in (\Gamma \setminus \mathcal{T})} \left \|\mathbf{e}_j ^T \mathbf{A}_{\Gamma} ^{\dagger} \mathbf{B}\right \|_2\right) \nonumber \\
	& = P\left(\esupmin_{i \in (\mathcal{T} \cap \Gamma)} \left \|\mathbf{e}_i ^T \mathbf{A}_{\Gamma} ^{\dagger} (\mathbf{A}_{\mathcal{T}}\mathbf{X}_\TColon+\mathbf{W})\right \|_2 > \esupmax_{j \in (\Gamma \setminus \mathcal{T})} \left \|\mathbf{e}_j ^T \mathbf{A}_{\Gamma} ^{\dagger}  (\mathbf{A}_{\mathcal{T}}\mathbf{X}_\TColon+\mathbf{W})\right \|_2 \right) \nonumber \\
	& > P\left(\esupmin_{i \in (\mathcal{T} \cap \Gamma)} \left( \left \|\mathbf{e}_i ^T \mathbf{A}_{\Gamma} ^{\dagger} \mathbf{A}_{\mathcal{T}}\mathbf{X}_\TColon\right \|_2 - \left \|\mathbf{e}_i ^T \mathbf{A}_{\Gamma} ^{\dagger} \mathbf{W}\right \|_2 \right) \nonumber \right.\\
	& \left. \quad \geq \esupmax_{j \in (\Gamma \setminus \mathcal{T})} \; \left( \left \|\mathbf{e}_j ^T \mathbf{A}_{\Gamma} ^{\dagger}  \mathbf{A}_{\mathcal{T}}\mathbf{X}_\TColon\right \|_2+\left \|\mathbf{e}_j ^T \mathbf{A}_{\Gamma} ^{\dagger}\mathbf{W}\right \|_2 \right) \right) \nonumber \\
	& (\text{Using reverse triangle inequality and triangle inequality respectively}) \nonumber \\
	& = P\Bigl(\esupmin_{i \in (\mathcal{T} \cap \Gamma)} \left \|\mathbf{e}_i ^T \mathbf{A}_{\Gamma} ^{\dagger} \mathbf{A}_{\mathcal{T}}\mathbf{X}_\TColon\right \|_2 - \esupmax_{j \in (\Gamma \setminus \mathcal{T})} \; \left \|\mathbf{e}_j ^T \mathbf{A}_{\Gamma} ^{\dagger}  \mathbf{A}_{\mathcal{T}}\mathbf{X}_\TColon\right \|_2 \nonumber \\
	& \quad > \esupmin_{i \in (\mathcal{T} \cap \Gamma)} \left \|\mathbf{e}_i ^T \mathbf{A}_{\Gamma} ^{\dagger}\mathbf{W}\right \|_2 +  \esupmax_{j \in (\Gamma \setminus \mathcal{T})} \left \|\mathbf{e}_j ^T \mathbf{A}_{\Gamma} ^{\dagger}\mathbf{W}\right \|_2 \Bigr) \nonumber \\
	& = P\left( \esupmin_{i \in (\mathcal{T} \cap \Gamma)} \left \|\mathbf{e}_i ^T \mathbf{A}_{\Gamma} ^{\dagger} \mathbf{A}_{\mathcal{T}}\mathbf{X}_\TColon\right \|_2 - \esupmax_{j \in (\Gamma \setminus \mathcal{T})} \; \left \|\mathbf{e}_j ^T \mathbf{A}_{\Gamma} ^{\dagger}  \mathbf{A}_{\mathcal{T}}\mathbf{X}_\TColon\right \|_2 > \eta \right) \nonumber\\
	& = 1- P\left( \esupmin_{i \in (\mathcal{T} \cap \Gamma)} \left \|\mathbf{e}_i ^T \mathbf{A}_{\Gamma} ^{\dagger} \mathbf{A}_{\mathcal{T}}\mathbf{X}_\TColon\right \|_2 - \esupmax_{j \in (\Gamma \setminus \mathcal{T})} \; \left \|\mathbf{e}_j ^T \mathbf{A}_{\Gamma} ^{\dagger}  \mathbf{A}_{\mathcal{T}}\mathbf{X}_\TColon\right \|_2 \leq \eta \right) \nonumber\\
	& \geq 1-P\left(\esupmin_{i \in (\mathcal{T} \cap \Gamma)} \left \|\mathbf{e}_i ^T \mathbf{A}_{\Gamma} ^{\dagger} \mathbf{A}_{\mathcal{T}}\mathbf{X}_\TColon\right \|_2 \leq C \right) \nonumber \\
	& \hspace{0.7cm} - P\left(\esupmax_{j \in (\Gamma \setminus \mathcal{T})} \; \left \|\mathbf{e}_j ^T \mathbf{A}_{\Gamma} ^{\dagger}  \mathbf{A}_{\mathcal{T}}\mathbf{X}_\TColon\right \|_2 \geq C- \eta\right). \label{tpr}
	\end{align}
	Now, let us derive an upper bound for  $\displaystyle P\left(\esupmin_{i \in (\mathcal{T} \cap \Gamma)} \left \|\mathbf{e}_i ^T \mathbf{A}_{\Gamma} ^{\dagger} \mathbf{A}_{\mathcal{T}}\mathbf{X}_\TColon\right \|_2 \leq C \right)$. Influenced by the concentration of measure results in \cite{Remi2008AtomsOfAll}, we set 
	\begin{align}
	\label{eqn:MMV_concInequality}
	C=(1-\epsilon _1) C_2(L) \esupmin_{i \in (\mathcal{T} \cap \Gamma)} \left \|\mathbf{e}_{i}^{T} \mathbf{A}_{\Gamma}^{\dagger} \mathbf{A}_{\mathcal{T}} \Sigma \right \|_2,
	\end{align} 
	where $0 < \epsilon _1 <1$.
	
	\noindent Using (5.5) in \cite{Remi2008AtomsOfAll}, we get,
	\begin{align}
	P\left(\esupmin_{i \in (\mathcal{T} \cap \Gamma)} \left \|\mathbf{e}_i ^T \mathbf{A}_{\Gamma} ^{\dagger} \mathbf{A}_{\mathcal{T}}\mathbf{X}_\TColon\right \|_2 \leq C \right) \leq |\mathcal{T}| \exp(-A_2(L) \epsilon _1^2). \label{fpr}
	\end{align}
	
	\noindent To bound the second probability, consider
	\begin{align*}
	& P\left(\esupmax_{j \in (\Gamma \setminus \mathcal{T})} \left \|\mathbf{e}_j ^T \mathbf{A}_{\Gamma} ^{\dagger} \mathbf{A}_{\mathcal{T}}\Sigma \Phi\right \|_2 \geq C-\eta \right) \\*
	& = P\left(\esupmax_{j \in (\Gamma \setminus \mathcal{T})} \left \|\mathbf{e}_j ^T \mathbf{A}_{\Gamma} ^{\dagger} \mathbf{A}_{\mathcal{T}}\Sigma \Phi\right \|_2 \geq (C-\eta) \frac{ C_2(L) \displaystyle \esupmax_{j \in (\Gamma \setminus \mathcal{T})} \left \|\mathbf{e}_{j}^{T} \mathbf{A}_{\Gamma}^{\dagger} \mathbf{A}_{\mathcal{T}} \Sigma \right \|_2}{C_2(L) \displaystyle \esupmax_{j \in (\Gamma \setminus \mathcal{T})} \left \|\mathbf{e}_{j}^{T} \mathbf{A}_{\Gamma}^{\dagger} \mathbf{A}_{\mathcal{T}} \Sigma \right \|_2} \right).
	\end{align*}
	Let
	\begin{align}
	\label{eqn:MMV_epsilon2}
	& 1+\epsilon_2=\frac{C-\eta}{C_2(L) \displaystyle \esupmax_{j \in (\Gamma \setminus \mathcal{T})} \left \|\mathbf{e}_{j}^{T} \mathbf{A}_{\Gamma}^{\dagger} \mathbf{A}_{\mathcal{T}} \Sigma \right \|_2}
	\end{align}
	
	\noindent Using equation (5.3) in \cite{Remi2008AtomsOfAll} 
	\begin{align}
	& P\left( \esupmax_{j \in (\Gamma \setminus \mathcal{T})} \left \|\mathbf{e}_j ^T \mathbf{A}_{\Gamma} ^{\dagger} \mathbf{A}_{\mathcal{T}}\Sigma \Phi\right \|_2 \geq (1+\epsilon _2) C_2(L) \esupmax_{j \in (\Gamma \setminus \mathcal{T})}\left \|\mathbf{e}_{j}^{T} \mathbf{A}_{\Gamma}^{\dagger} \mathbf{A}_{\mathcal{T}} \Sigma \right \|_2 \right) \nonumber \\
	&\hspace{3cm} \leq |\mathcal{T}| \exp\left(-A_2(L) \epsilon _2^2\right). \label{spr}
	\end{align}

	\noindent For the above inequality to hold, it is required that $\epsilon _2 > 0$. 
	By setting $\epsilon_2 = \epsilon_1$, and using \eqref{eqn:MMV_concInequality} and \eqref{eqn:MMV_epsilon2}, we get
	\begin{align*}
	\epsilon_1 = \frac{(1-\epsilon_1)C_2(L) \min_{i\in (\T \cap \Gamma)} \vecLTwoNorm{\e_i^T \A_\Gamma^\dagger \A_\T \Sigma}-\eta }{C_2(L)\displaystyle\max_{j\in(\Gamma\setminus\T)}\vecLTwoNorm{\e_j^T \A_\Gamma^\dagger \A_\T \Sigma}} -1.
	\end{align*}
	Now, solving for $\epsilon$, we get

	\begin{align*}
	\epsilon_1 = \frac{\displaystyle \min_{i\in (\T \cap \Gamma)} \vecLTwoNorm{\e_i^T \A_\Gamma^\dagger \A_\T \Sigma}-\displaystyle\max_{j\in(\Gamma\setminus\T)}\vecLTwoNorm{\e_j^T \A_\Gamma^\dagger \A_\T \Sigma} - \frac{\eta}{C_2(L)}    }  {\displaystyle \min_{i\in (\T \cap \Gamma)} \vecLTwoNorm{\e_i^T \A_\Gamma^\dagger \A_\T \Sigma} + \displaystyle\max_{j\in(\Gamma\setminus\T)}\vecLTwoNorm{\e_j^T \A_\Gamma^\dagger \A_\T \Sigma} }.	
	\end{align*}
	
	\noindent Clearly $\epsilon _1 < 1$ and by the assumption in the theorem $\epsilon _1 > 0$. Hence we have $0 < \epsilon_1 < 1$. Also, note that $\gamma = \epsilon _1$. Substituting \eqref{fpr} and \eqref{spr} in \eqref{tpr}, we get 
	\begin{align*}
	P(\Theta) \geq 1-K \exp (-2 A_2(L) \gamma ^2).
	\end{align*}
\end{proof}
\noindent Since $A_2(L)\approx \displaystyle \frac{L}{2}$, the probability that {MMV-FACS} selects all correct indices from the union set increases as $L$ increases. Thus, more than one measurement vector improves the performance.

\section{Numerical Experiments and Results}
We conducted numerical experiments using synthetic data and real signals to evaluate the performance of {MMV-FACS}. The performance is evaluated using {ASRER} which is defined as 
\begin{align}
\label{eqn:MMV_ASRER}
\text{{ASRER}}= \frac{\sum _{j=1}^{n_{trials}}\left \|\mathbf{X}_j\right \|_F^2 }{\sum _{j=1}^{n_{trials}} \left \|\mathbf{X}_j-\hat{\mathbf{X}}_j\right \|_F^2 },
\end{align}
where $\mathbf{X}_j$ and $\hat{\mathbf{X}}_j$ denote the actual and reconstructed jointly sparse signal matrix in the $j^\text{th}$ trial respectively, and $n_{trials}$ denotes the total number of trials. 

\subsection{Synthetic Sparse Signals}
For noisy measurement simulations, we define the {SMNR} as 
\begin{equation*}
\text{{SMNR}} \triangleq \frac{\E\{ \left \|\mathbf{x}^{(i)}\right \|_2^2\}}{\E \{ \left \|\mathbf{w}^{(i)}\right \|_2^2\}},
\end{equation*}
where $\E\{\cdot\}$ denotes the mathematical expectation operator. The simulation set-up is described below.
\subsubsection{Experimental Setup} 
\label{cases}
Following steps are involved in the simulation: 
\begin{enumerate}
	\item Generate elements of $\mathbf{A}_{M \times N}$ independently from $\mathcal{N}(0,\frac{1}{M})$ and normalize each column norm to unity. 
	\item Choose $K$ non-zero locations uniformly at random from the set $\{1,2,\ldots,N\}$ and fill those K rows of $\mathbf{X}$ based on the choice of signal characteristics:
	\begin{enumerate} 
		\item \emph{Gaussian sparse signal matrix}: Non-zero values independently from $\mathcal{N}(0,1)$. 	
		\item \emph{Rademacher sparse signal matrix}: Non-zero values are set to +1 or -1 with probability $\frac{1}{2}$. 
	\end{enumerate}
	Remaining $N-K$ rows of $\mathbf{X}$ are set to zero.
	\item The {MMV} measurement matrix $\mathbf{B}$ is computed as $\mathbf{B}=\mathbf{A}\mathbf{X}+\mathbf{W}$, where the columns of $\mathbf{W}$, $\mathbf{w}^{(i)}$'s are independent and their elements are {i.i.d.} as Gaussian with variance determined from the specified {SMNR}.
	\item Apply the {MMV} sparse recovery method.
	\item Repeat steps i-iv, $S$ times. 
	\item Find {ASRER} using \eqref{eqn:MMV_ASRER}.
\end{enumerate}
\begin{figure}[!htb]
	\centering
	\subfigure[]
	{
		\includegraphics[width=2.75in]{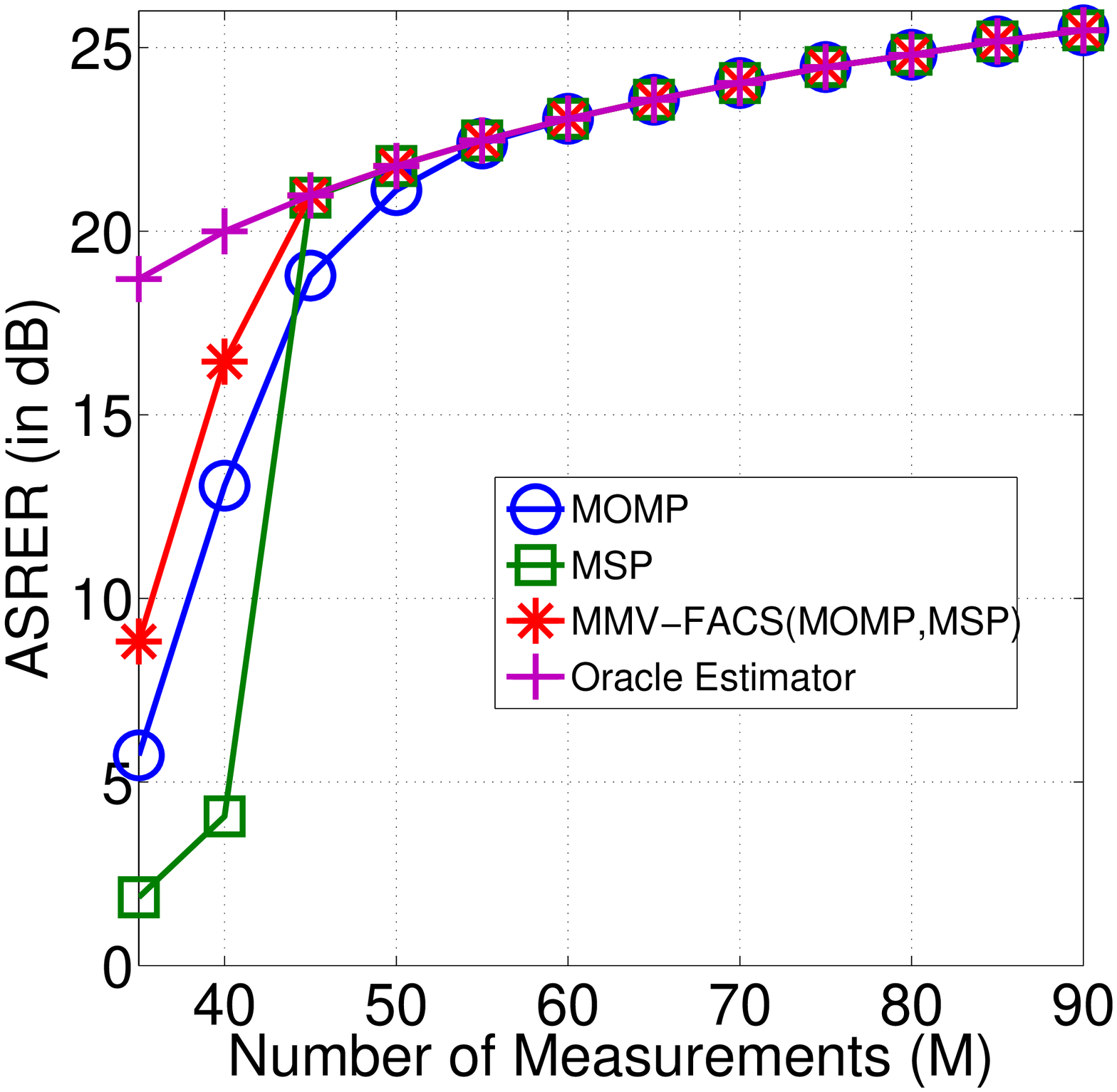}
	}
	\subfigure[]
	{
		\includegraphics[width=2.75in]{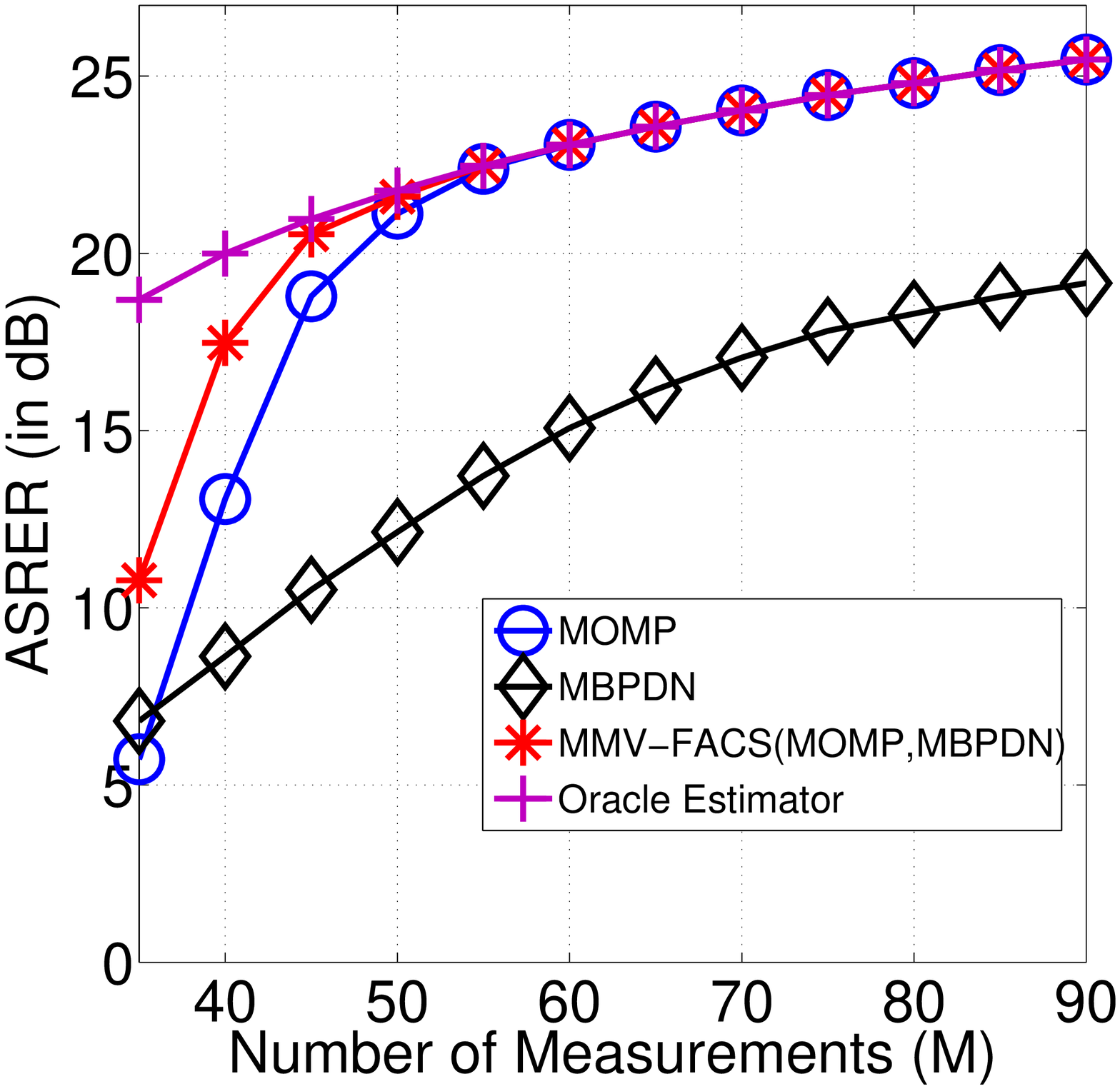}
	} \\
	\subfigure[]
	{
		\includegraphics[width=2.75in]{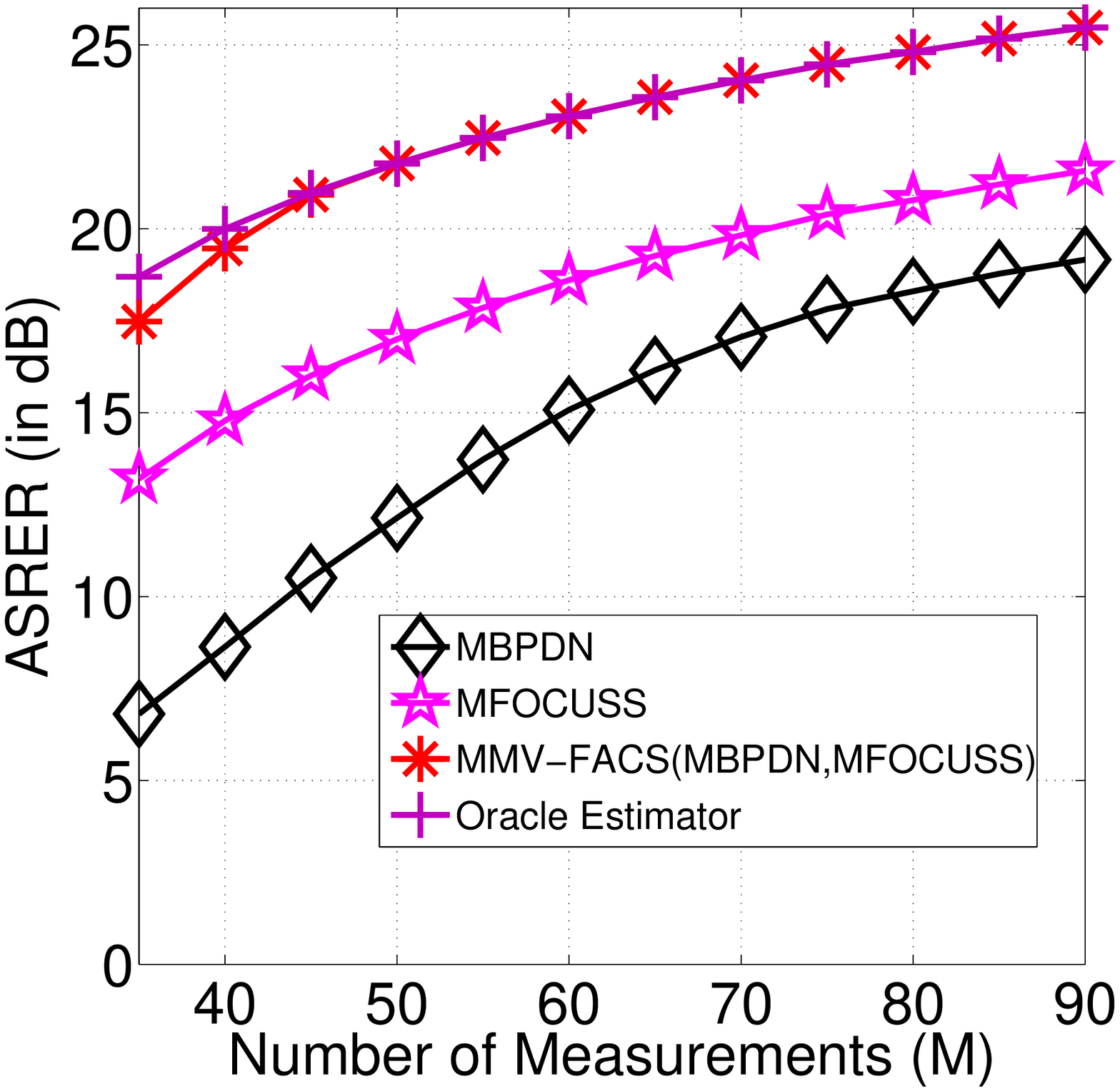}
	}
	\subfigure[]
	{
		\includegraphics[width=2.75in]{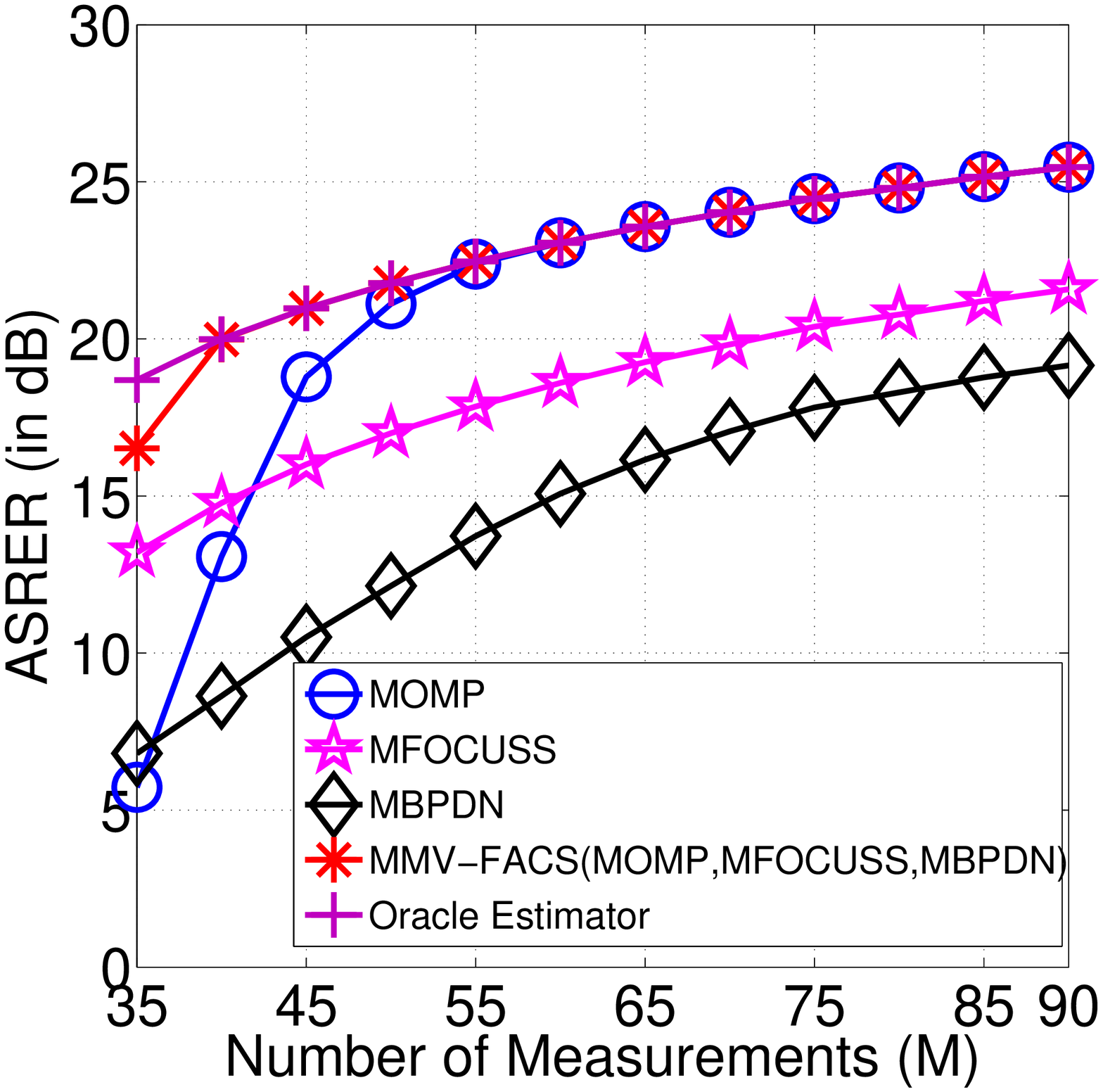}
	}
	\caption[Performance of MMV-FACS for Gaussian sparse signal matrices, varying the number of measurements ($M$)]{Performance of MMV-FACS, averaged over $1,000$ trials, for Gaussian sparse signal matrices with $\text{{SMNR}}=20$~dB. Sparse signal dimension $N=500$, sparsity level $K=20$, and number of measurement vectors $L=20$.}
	\label{Fig:one}
\end{figure}
\subsubsection{Results and Discussions}
\begin{figure}[!htb]
	\centering
	\subfigure[]
	{
		\includegraphics[width=2.75in]{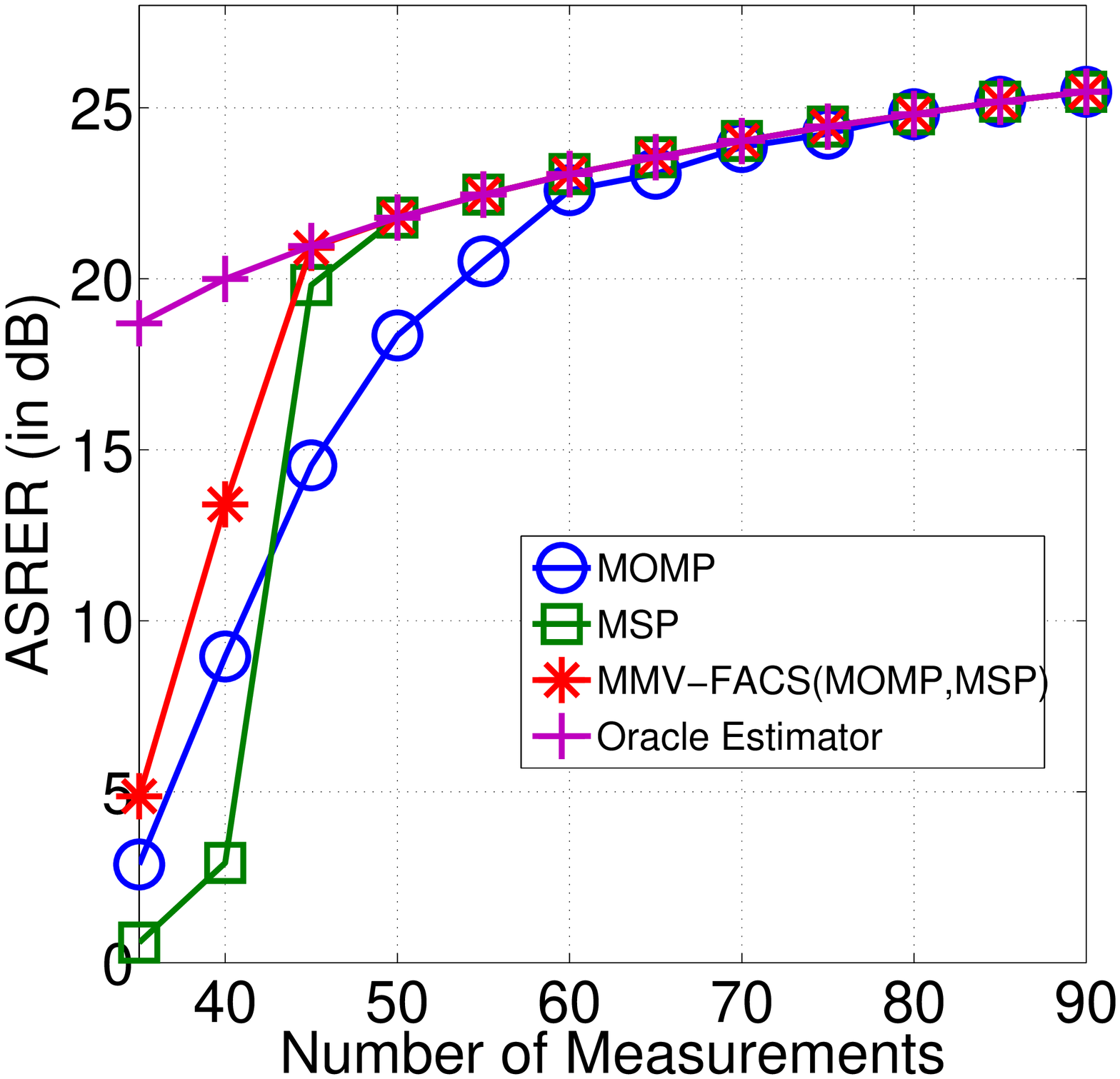}
	}
	\subfigure[]
	{
		\includegraphics[width=2.75in]{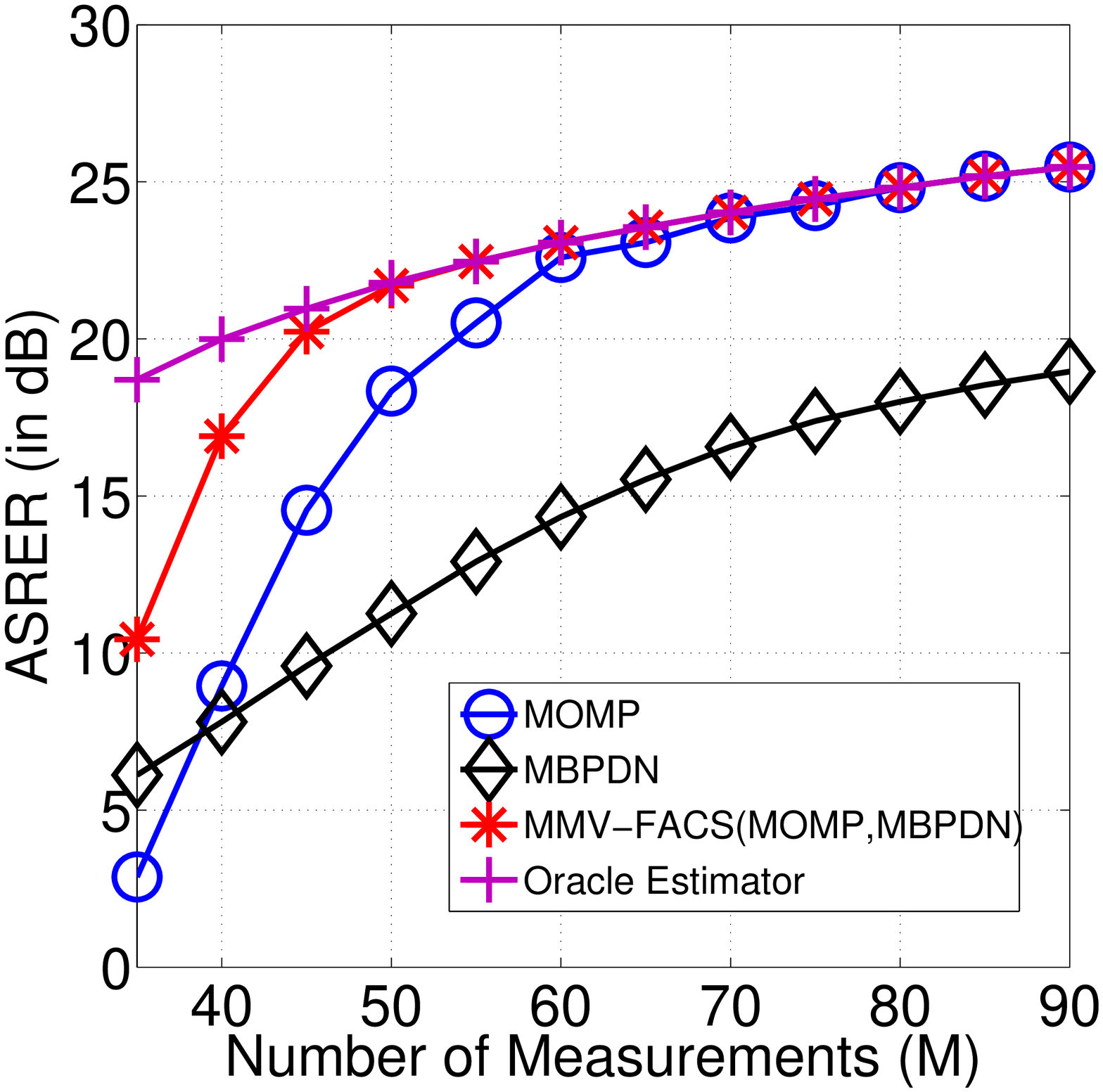}
	} \\
	\subfigure[]
	{
		\includegraphics[width=2.75in]{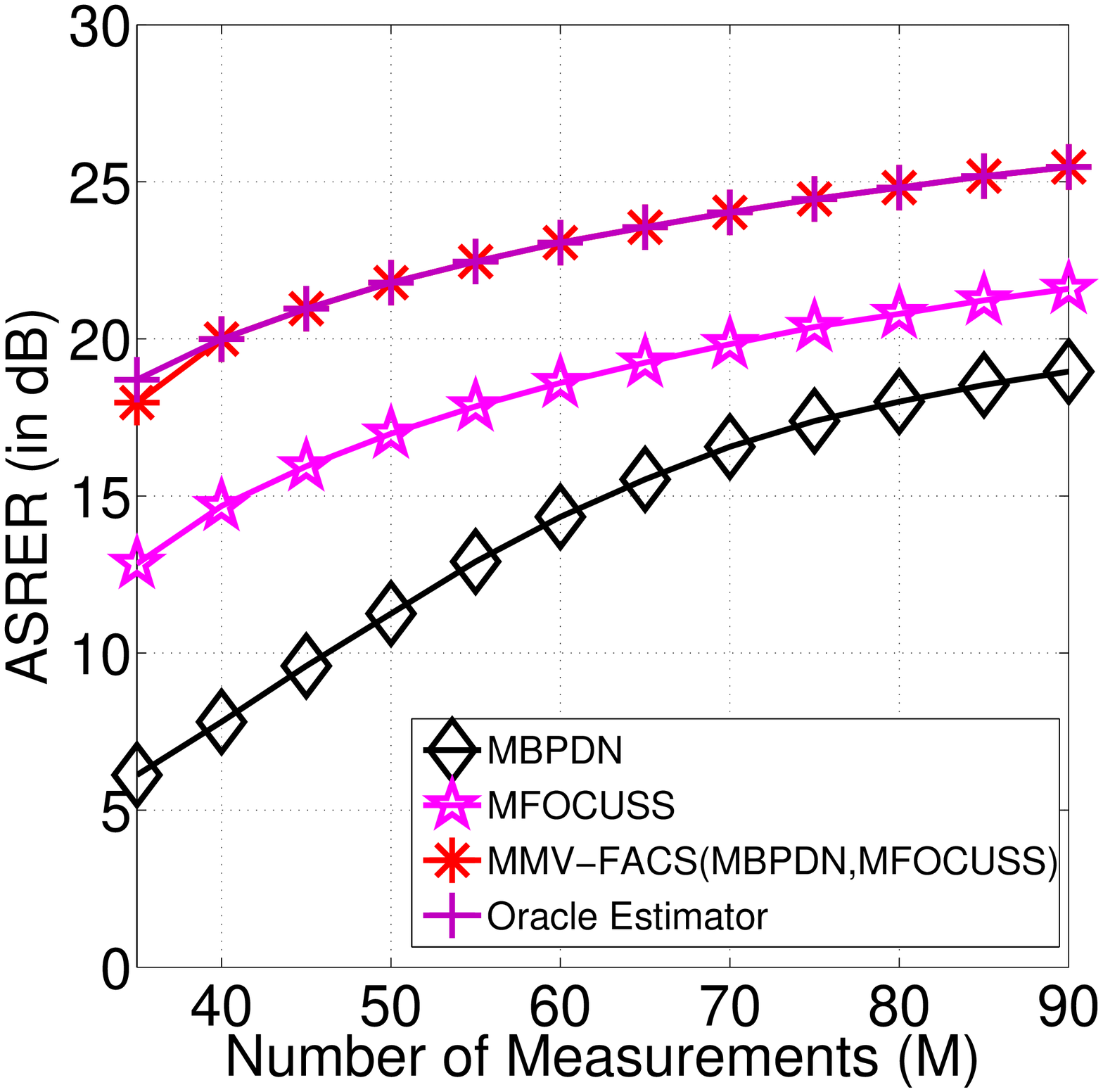}
	}
	\subfigure[]
	{
		\includegraphics[width=2.75in]{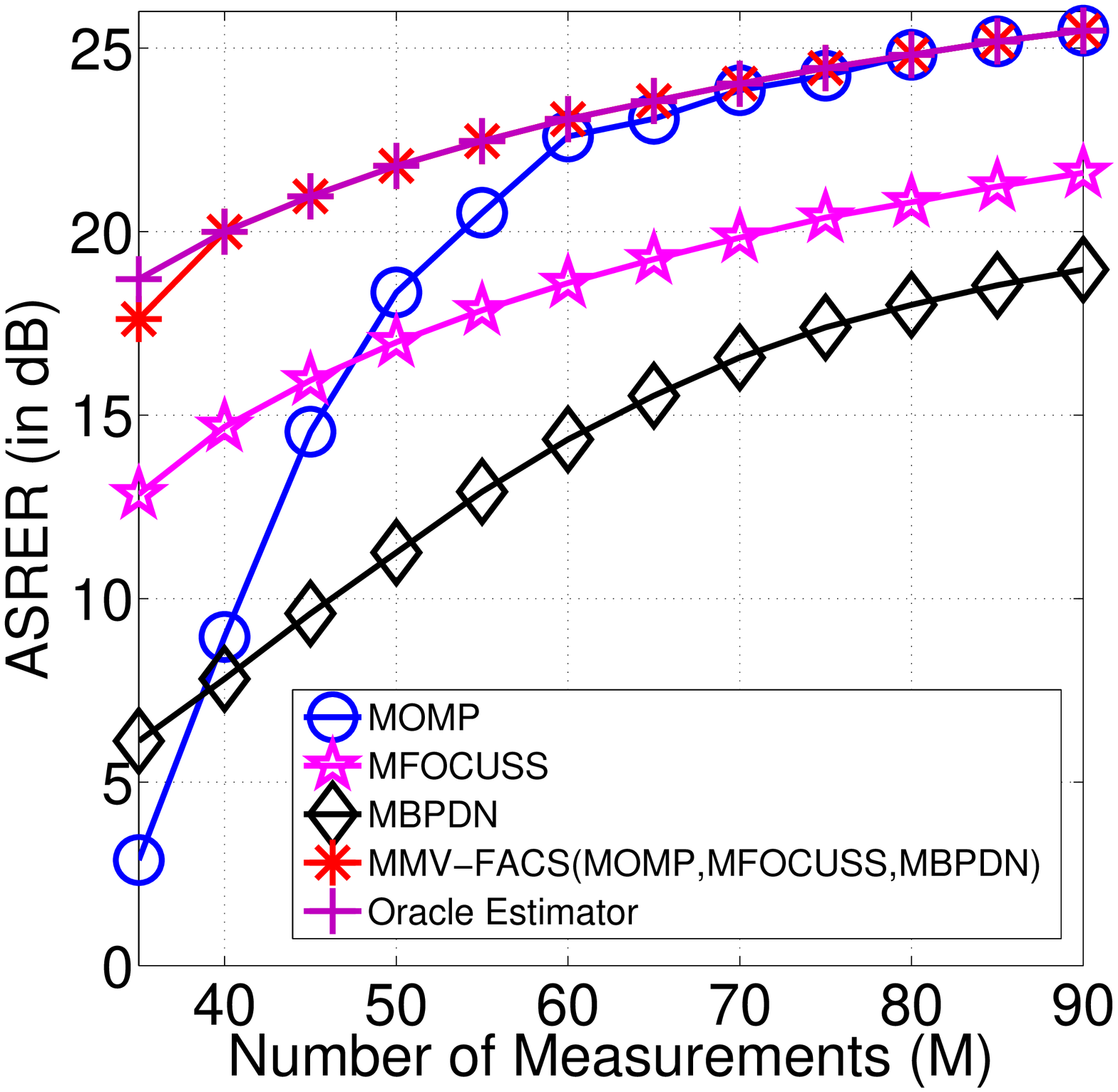}
	}
	\caption[Performance of MMV-FACS for Rademacher sparse signal matrices, varying the number of measurements ($M$)]{Performance of MMV-FACS, averaged over $1000$ trials, for Rademacher sparse signal matrices with $\text{{SMNR}}=20$~dB. Sparse signal dimension $N=500$, sparsity level $K=20$ and number of measurement vectors $L = 20$.}
	\label{Fig:two}
\end{figure}

\begin{figure}[!htb]
	\centering
	\subfigure[]
	{
		\includegraphics[width=2.75in]{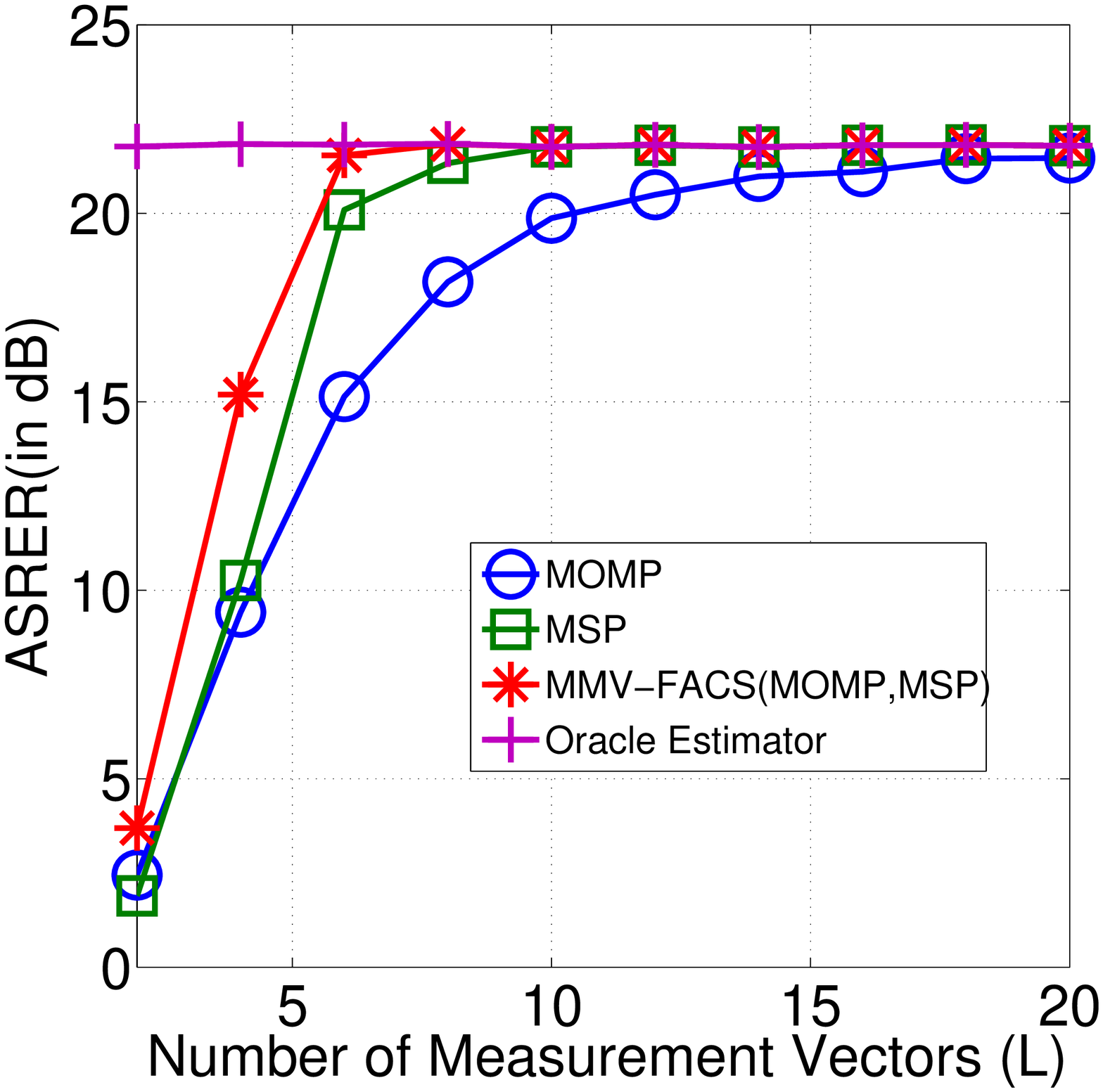}
	}
	\subfigure[]
	{
		\includegraphics[width=2.75in]{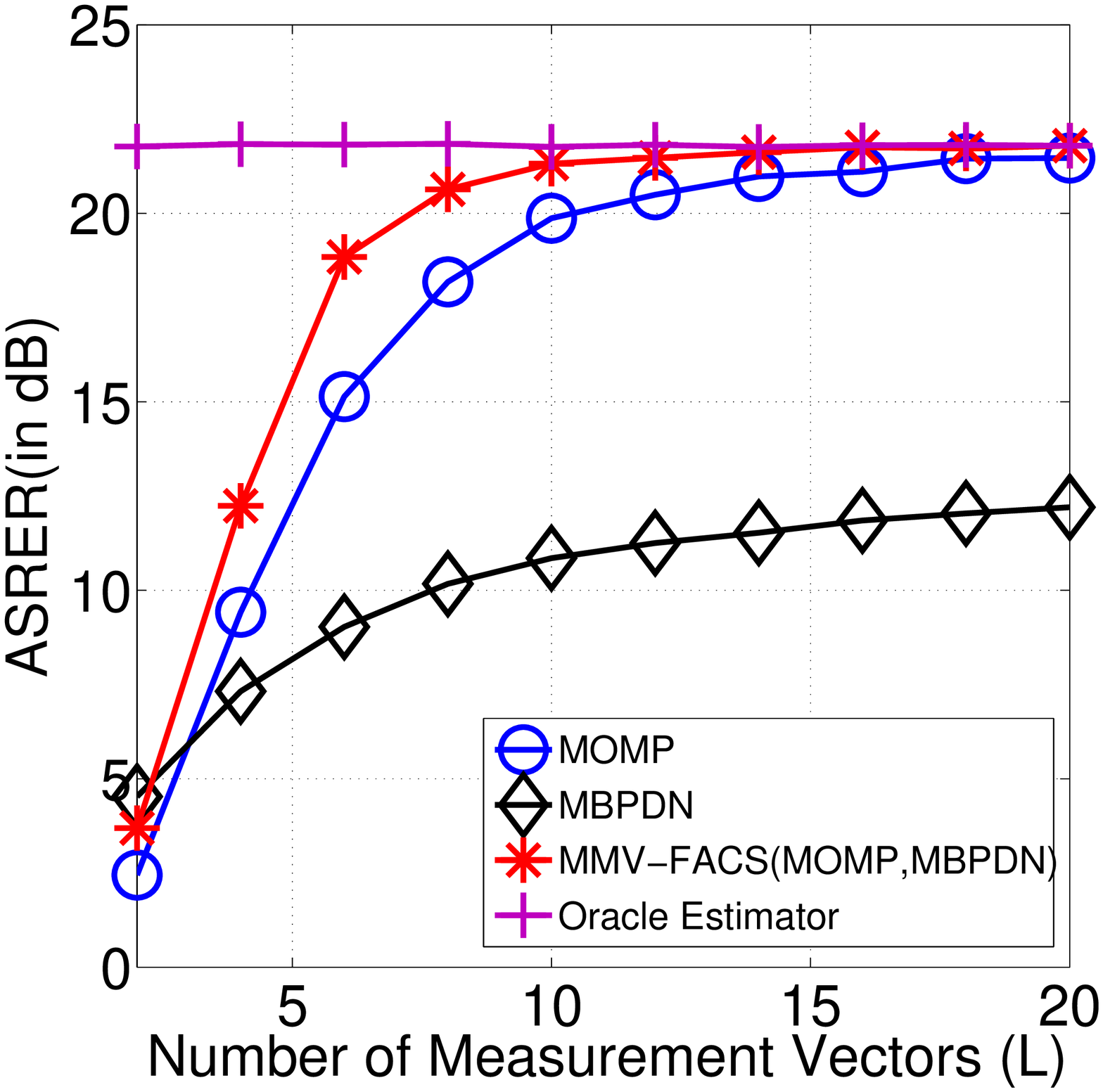}
	} 
	\caption[Performance of MMV-FACS for Gaussian sparse signal matrices, varying the number of measurement vectors ($L$)]{Performance of MMV-FACS, averaged over $1000$ trials, for Gaussian sparse signal matrices with $\text{{SMNR}}=20$~dB. Sparse signal dimension $N = 500$, sparsity level $K = 20$, and number of measurements $M=50$.}	
	\label{Fig:three}
\end{figure}
\begin{figure}[!htb]
	\centering
	\subfigure[]
	{\includegraphics[width=2.75in]{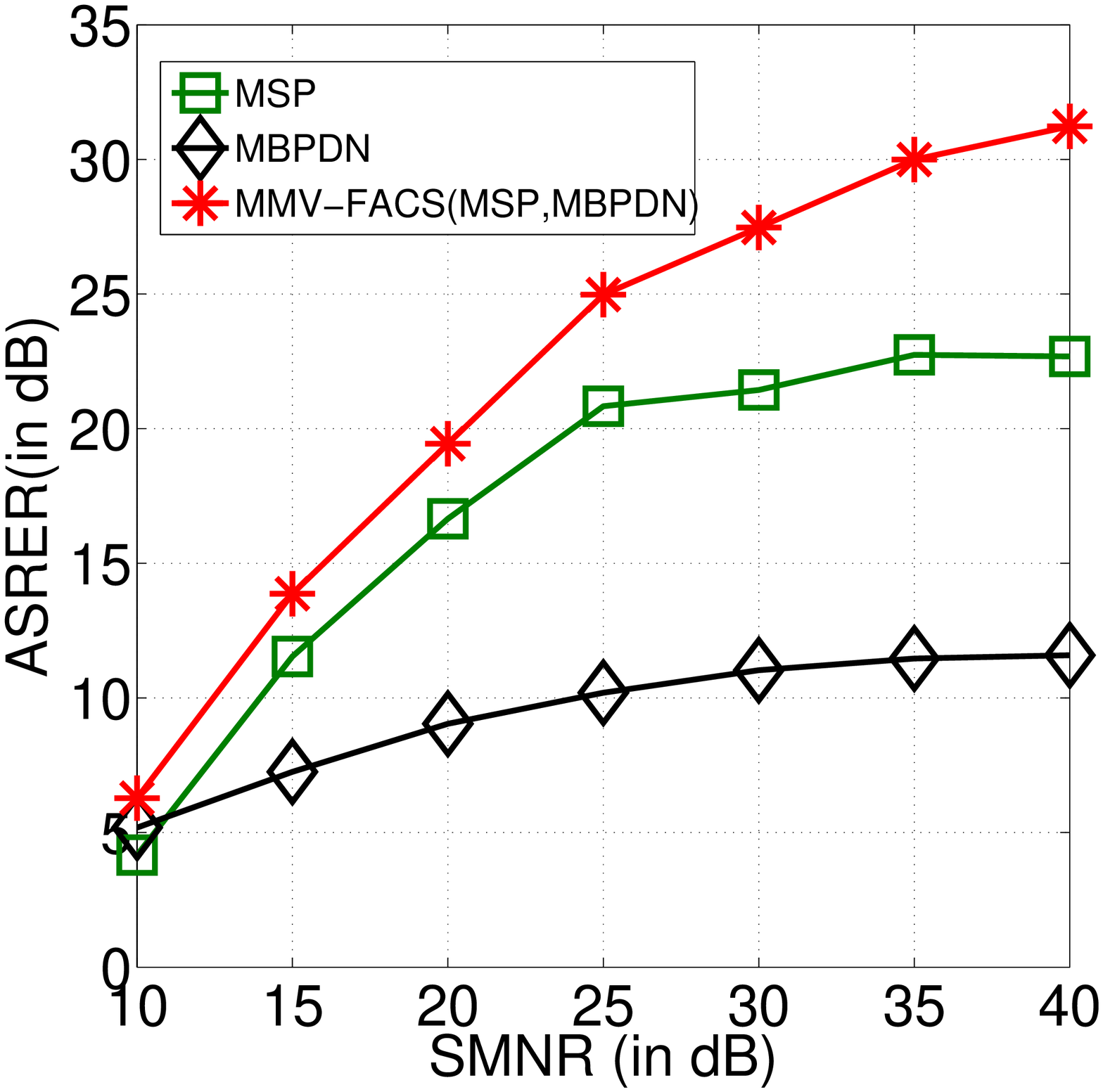}
	}
	\subfigure[]
	{
		\includegraphics[width=2.75in]{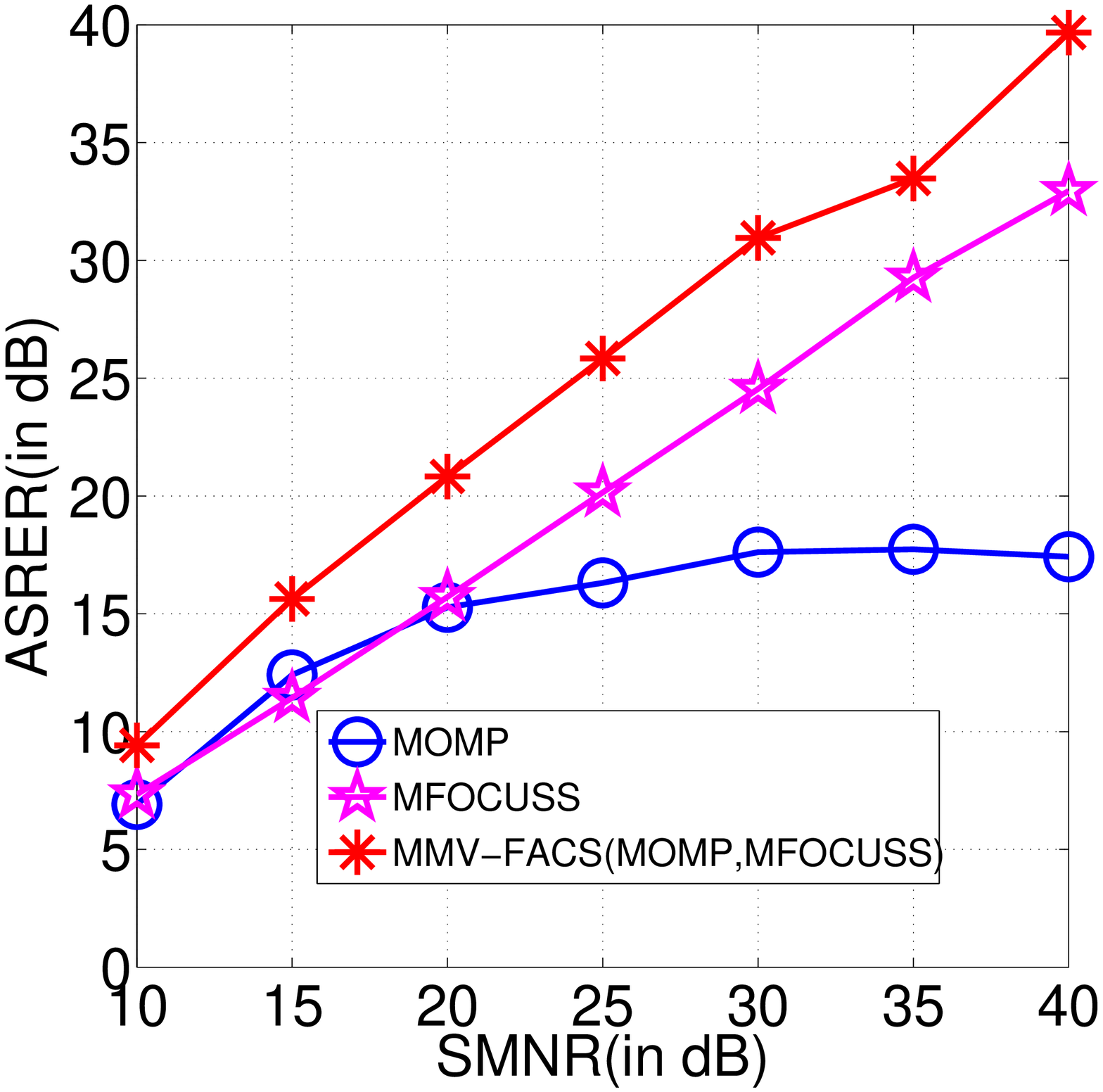}
	} \\
	\subfigure[]
	{
		\includegraphics[width=2.75in]{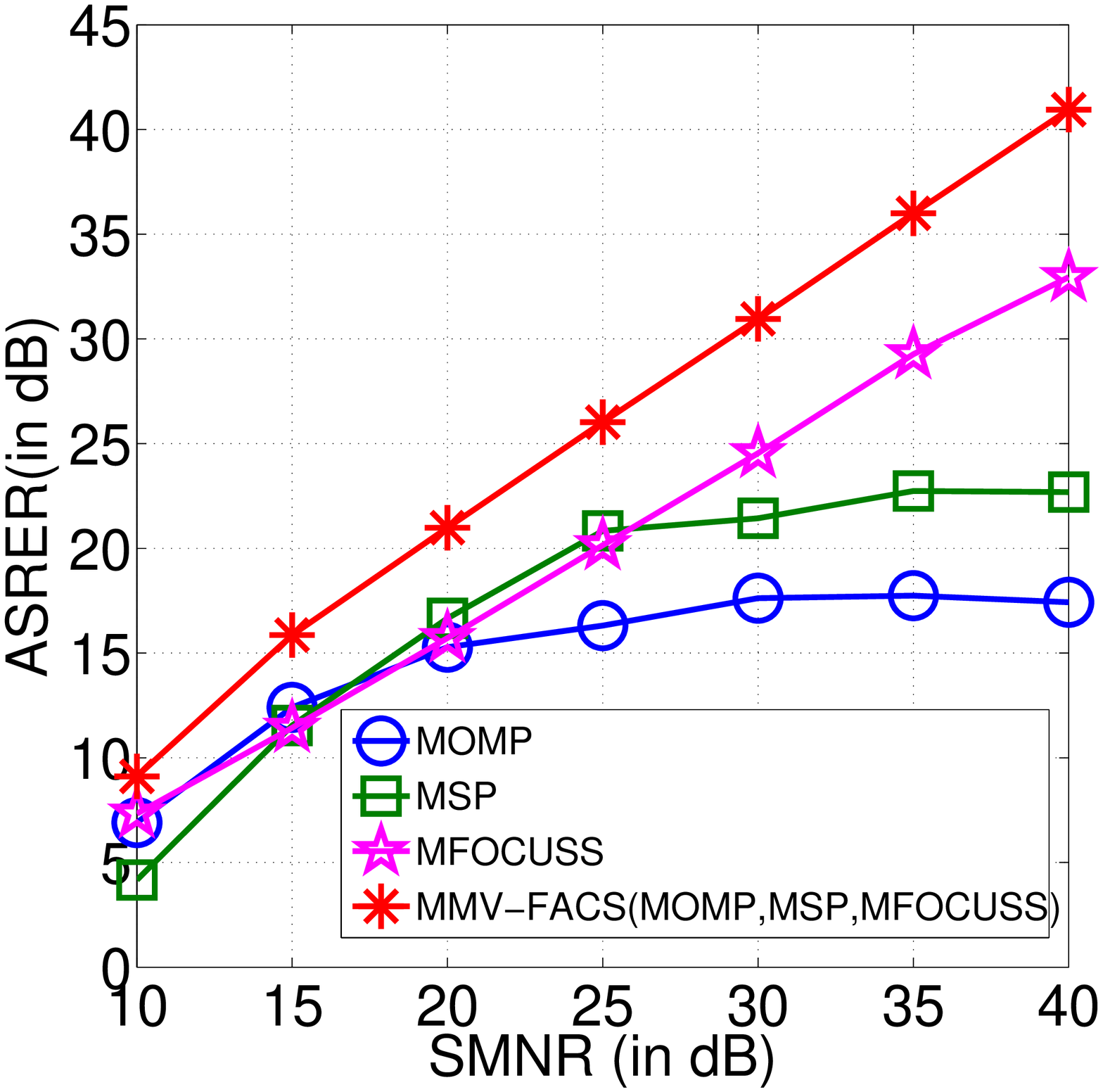}
	}
	\subfigure[]
	{
		\includegraphics[width=2.75in]{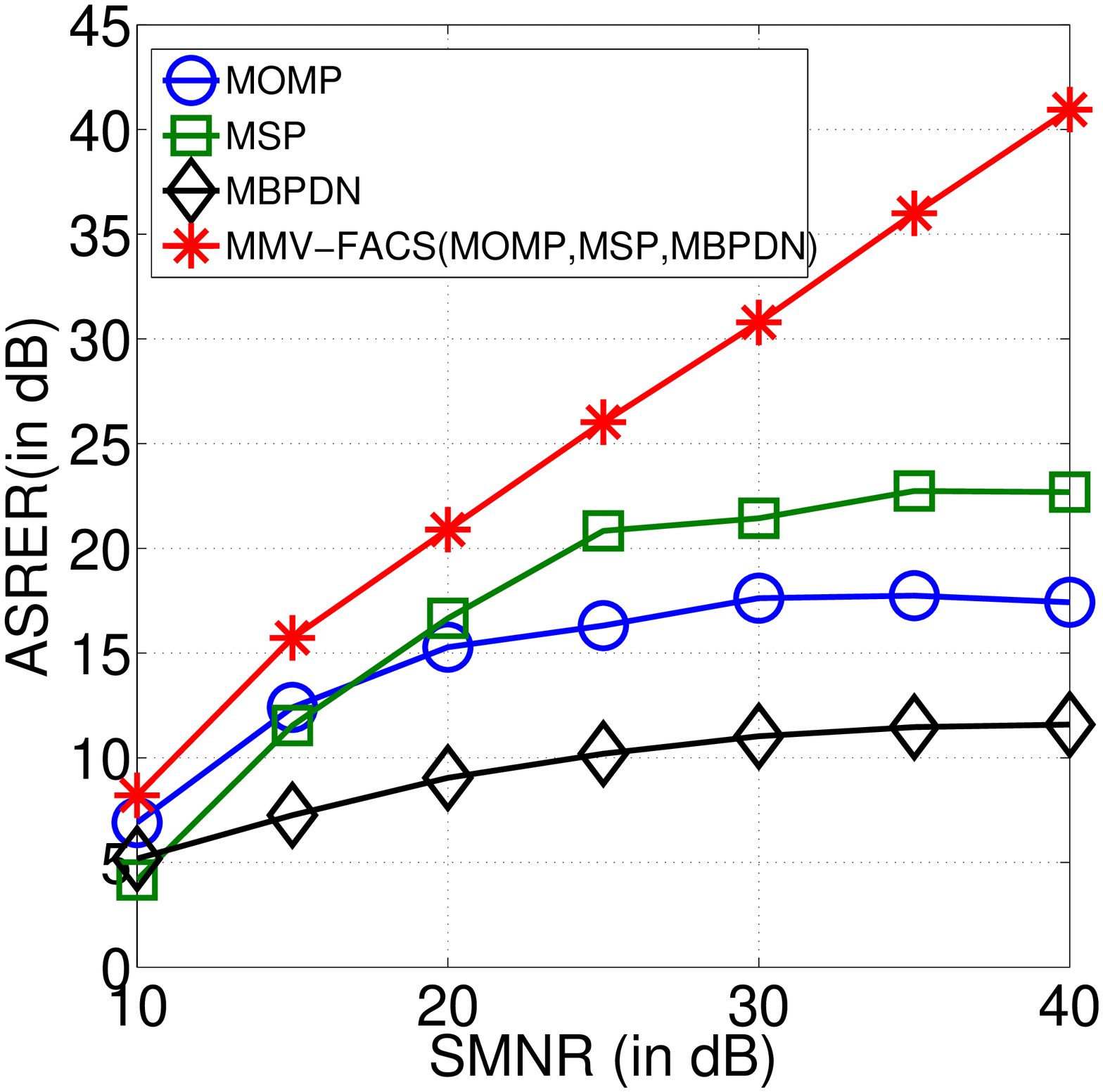}
	}
	\caption[Performance of MMV-FACS for Gaussian sparse signal matrices, varying SMNR]{Performance of MMV-FACS, averaged over $1000$ trials, for Gaussian sparse signal matrices with $\text{{SMNR}}=20$~dB. Sparse signal dimension $N = 500$, sparsity level $K = 20$, and number of measurements $M=45$, and number of measurement vectors $L=10$.}
	\label{Fig:four}
\end{figure}

We used M-OMP, M-SP, M-BPDN \cite{Berg2009Probing}, and M-FOCUSS \cite{Cotter2005SparseSolution} as the participating algorithms in {MMV-FACS}. The software code for M-BPDN was taken from SPGL1 software package~\cite{SPGL12007}. Since M-FOCUSS and M-BPDN algorithms may not yield an exact $K$-sparse solution, we estimate the support-set as the indices of the $K$ rows with largest $\ell_2$ norm.  We fixed the sparse signal dimension $N=500$ and sparsity level $K=20$ in the simulation the result were calculated by averaging over $1,000$ trials ($S = 1,000$). We use an \emph{oracle estimator} for performance benchmarking. The \emph{oracle estimator} is aware of the true support-set and finds the non-zero entries of the sparse matrix by solving {LS}.

The empirical performance of {MMV} reconstruction algorithms for different values of $M$ is shown in \figurename{~\ref{Fig:one}}. The simulation parameters are $L=20$, {SMNR}$=20$ dB and \textbf{X} is chosen as Gaussian sparse signal matrix. For $M=35$, {MMV-FACS} (M-BPDN,M-FOCUSS) gave 10.67 dB and 4.27 dB improvement over M-BPDN and M-FOCUSS respectively.\\

\figurename{~\ref{Fig:two}} depicts the performance of Rademacher sparse signal matrix for different values of $M$ where we set $L=20$ and {SMNR}$=20$ dB. We again observe similar performance improvement as in the case of Gaussian sparse signal matrix. For example, for $M=35$, {MMV-FACS}(M-OMP,M-BPDN) showed 7.56 dB and 4.32 dB over M-OMP and M-BPDN respectively.

A comparison of {MMV} reconstruction techniques is shown in \figurename{~\ref{Fig:three}}for Gaussian sparse signal matrix for different values of $L$ where we set $M=50$ and {SMNR}$=20$ dB. It may be observed that {MMV-FACS} gave a significant performance improvement over the participating algorithms. Specifically, MMV-FACS(M-OMP,M-SP) improved the performance by $5.77$~dB and $4.94$~dB over M-OMP and M-SP respectively.

To show the dependency of recovery performance on {SMNR}, we conducted simulations for different values of {SMNR}. \figurename{~\ref{Fig:four}} illustrates the performance for Gaussian sparse signal matrix where $L=10$ and $M=45$. An additional {ASRER} improvement of $2.51$~dB and $2.08$~ dB were achieved as compared to M-OMP and M-FOCUSS respectively for {SMNR}$=10$~dB. This shows the robustness of {MMV-FACS} to noisy measurements.

From the above simulation results it can be seen that {MMV-FACS} improved the sparse signal recovery compared to participating algorithms.
\subsubsection{Reproducible Research}
We provide necessary Matlab codes to reproduce all the figures, publicly downloadable from \url{http://www.ece.iisc.ernet.in/~ssplab/Public/MMVFACS.tar.gz}.

\subsection{Real Compressible Signals}
To evaluate the performance of {MMV-FACS} on compressible signals and real world data, we used the data set `$\mathit{05091.dat}$' from  MIT-BIH Atrial Fibrillation Database~\cite{Goldberger2000PhysioBank}. The recording is of $10$ hours in duration, and contains two ECG signals each sampled at $250$ samples per second with $12$-bit resolution over a range of $\pm 10$ millivolts. We selected the first $250$ time points of the recording as the data set used in our experiment. We used a randomly generated Gaussian sensing matrix of size $M \times 250$, with different values of $M$ in the experiment. We assumed sparsity level $K=50$ and used  M-OMP and M-SP as the participating algorithms. The reconstruction results are shown in \figurename{~\ref{Fig:realdata}}.
\begin{figure}[!t]
	\centering
	\includegraphics[width = 0.49\textwidth]{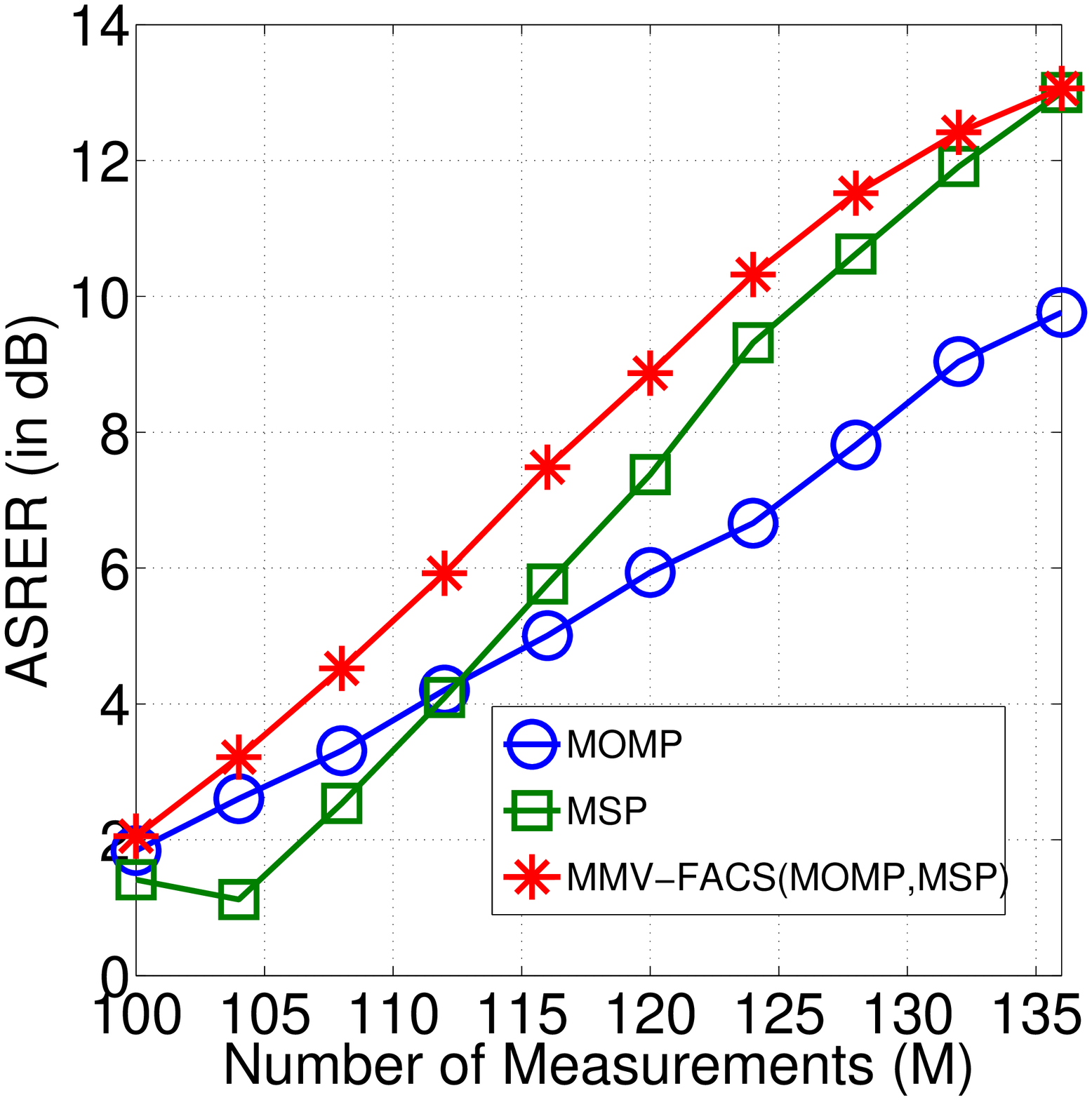}
	\caption[Performance of {MMV-FACS} for real compressible signals]{Real Compressible signals: Performance of {MMV-FACS} for 2-channel ECG signals from MIT-BIH Atrial Fibrillation Database~\cite{Goldberger2000PhysioBank}.}
	\label{Fig:realdata}
\end{figure}

Similar to synthetic signals, {MMV-FACS} shows a better {ASRER} compared to the participating algorithms M-OMP and M-SP. This demonstrates the advantage of {MMV-FACS} in real-life applications, requiring fewer measurement samples to yield an approximate reconstruction. 

\section{Conclusions}
In this paper, we extended {FACS} to the {MMV} case and showed that {MMV-FACS} improves sparse signal matrix reconstruction. Using RIP, we theoretically analysed the proposed scheme and derived sufficient conditions for the performance improvement over the participating algorithm. Using Monte-Carlo simulations, we showed the performance improvement of the proposed scheme over the participating methods. Though this paper discusses only the extension of {FACS} for {MMV} problem, a similar approach can be used to extend the other fusion algorithms developed by Ambat {\em et al.}\ ~\cite{Ambat2012FuGP_EUSIPCO,Ambat2014CoMACS_IEEE_TSP,Ambat2014ProgresiveFusion}.
\begin{appendix}
	
	\section*{Proof of Lemma~\ref{lem1_MMVFACS}}
	\label{app:appendix1}
	\noindent The proof is inspired by Proposition~3.5 by Needell and Tropp~\cite{Needell2009CoSaMP}.\\
	Define set S as the convex combination of all matrices which are $R+K$ sparse and have unit Frobenius norm. 
	\begin{equation*}
	S=conv\Biggl\{ \mathbf{X}: \left \|\mathbf{X}\right \|_0 \leq R+K, \left \|\mathbf{X}\right \|_F=1 \Biggr\} 
	\end{equation*}
	Using the relation $\left \|\mathbf{AX}\right \|_F \leq \sqrt{1+\delta _{R+K}} \left \|\mathbf{X}\right \|_F$, we get,
	\begin{align*}
	\esupa_{\mathbf{X} \in S} \left \|\mathbf{AX}\right \|_F \leq & \sqrt{1+\delta _{R+K}}. 
	\end {align*}
	
	\noindent Define  
	\begin{equation*}
	Q=\Biggl\{ \mathbf{X}: \left \|\mathbf{X}\right \|_F + \frac{1}{\sqrt{R+K}} \left \|\mathbf{X}\right \|_{2,1} \leq 1 \Biggr\}. 
	\end{equation*}
	
	\noindent The lemma essentially claims that 
	\begin{equation*}
	\esupa_{\mathbf{X} \in Q} \left \|\mathbf{AX}\right \|_F \leq \esupa_{\mathbf{X} \in S} \left \|\mathbf{AX}\right \|_F. 
	\end {equation*}
	
	\noindent  To prove this, it is sufficient to ensure that $Q \subset S$.
	
	\noindent  Consider a matrix $\mathbf{X} \in Q$. Partition the support of $\mathbf{X}$ into sets of size $R+K$. Let set $\mathcal{I}_0$ contain the indices of the $R+K$ rows of $\mathbf{X}$ which have largest row $\ell_2$-norm, breaking ties lexicographically. Let set $\mathcal{I}_1$ contain the indices of the next largest (row $\ell_2$-norm) $R+K$ rows and so on. The final block $\mathcal{I}_J$ may have lesser than $R+K$ components. This partition gives rise to the following decomposition: 
	\begin{equation*}
	\mathbf{X}=\mathbf{X}|_{I_0}+\sum_{j=1}^{J}\mathbf{X}|_{I_j} = \lambda _0 \mathbf{Y}_0 + \sum_{j=1}^{J} \lambda _j \mathbf{Y}_j,
	\end{equation*}
	where  $\displaystyle \lambda _j =\left \|\mathbf{X}|_{I_j}\right \|_F \quad \text{and} \quad \mathbf{Y}_j= \lambda_j^{-1} \mathbf{X}|_{I_j}.$
	
	\noindent By construction each matrix $\mathbf{Y}_j$ belongs to $S$ because it is $R+K$ sparse and has unit Frobenius norm. We will show that $\sum_{j} \lambda _j \leq 1$. This implies that $\mathbf{X}$ can be written as a convex combination of matrices from the set $S$. As a result $\mathbf{X} \in S$. Therefore, $Q\subset S$. 
	
	\noindent Fix some $j$ in the range $\{1,2, \ldots , J\}$. Then, $\mathcal{I}_j$ contains at most $R+K$ elements and $\mathcal{I}_{j-1}$ contains exactly $R+K$ elements. Therefore,
	\begin{equation*}
	\lambda _j = \left \|\mathbf{X}|_{I_j}\right \|_F \leq \sqrt{R+K} \left \|\mathbf{X}|_{I_j}\right \|_{2,\infty} \leq \sqrt{R+K} \cdot \frac{1}{R+K} \left \|\mathbf{X}|_{I_{j-1}}\right \|_{2,1}.
	\end{equation*}
	The last inequality holds because the row $\ell_2$-norm of $\mathbf{X}$ on the set $I_{j-1}$ dominates its largest row $\ell_2$-norm in $\mathcal{I}_j$. Summing these relations, we get
	\begin{equation*}
	\sum _{j=1}^{J} \lambda _j \leq \frac{1}{\sqrt{R+K}} \sum _{j=1}^J \left \|\mathbf{X}|_{I_{j-1}}\right \|_{2,1} \leq \frac{1}{\sqrt{R+K}} \left \|\mathbf{X}\right \|_{2,1}. 
	\end{equation*}
	
	\noindent Also, we have $\lambda _0 = \left \|\mathbf{X}|_{I_0}\right \|_F \leq \left \|\mathbf{X}\right \|_F$. Since  $\mathbf{X} \in Q$, we conclude that
	\begin{equation*}
	\sum _{j=0}^{J} \lambda _j \leq \Biggl[ \left \|\mathbf{X}\right \|_F + \frac{1}{\sqrt{R+K}} \left \|\mathbf{X}\right \|_{2,1}\Biggr] \leq 1. \tag*{$\blacksquare$} 
	\end{equation*}
	
	\section*{Proof of Lemma~\ref{lem2_MMVFACS}}
	\label{app:appendix2}
	Let $\mathbf{y}^{(i)}$ denote the $i^\text{th}$ column of matrix $\mathbf{Y}$ and $\mathbf{r}^{(i)}$ denote the $i^\text{th}$ column of matrix $\mathbf{R}$, $i=1,2,\ldots L$. Then we have from Lemma~2 of \cite{WeiDai2009SubspacePursuit}
	\begin{align*}
	\left(1-\frac{\delta _{|\mathcal{T}_1|+|\mathcal{T}_2|}}{1-\delta _{|\mathcal{T}_1|+|\mathcal{T}_2|}} \right) \left \|\mathbf{y}^{(i)}\right \|_2 & \leq \left \|\mathbf{r}^{(i)}\right \|_2  \leq \left \|\mathbf{y}^{(i)}\right \|_2.
	\end{align*}
	
	\noindent  Summing the above relation, we obtain
	\begin{align*}
	\left(1-\frac{\delta _{|\mathcal{T}_1|+|\mathcal{T}_2|}}{1-\delta _{|\mathcal{T}_1|+|\mathcal{T}_2|}} \right) \sum _{i=1}^{L}\left \|\mathbf{y}^{(i)}\right \|_2^2 & \leq \sum _{i=1}^L\left \|\mathbf{r}^{(i)}\right \|_2^2 \leq \sum _{i=1}^L\left \|\mathbf{y}^{(i)}\right \|_2^2. 
	\end{align*}
	
	\noindent Equivalently, we have,
	\begin{align*}
	\left(1-\frac{\delta _{|\mathcal{T}_1|+|\mathcal{T}_2|}}{1-\delta _{|\mathcal{T}_1|+|\mathcal{T}_2|}} \right) \left \|\mathbf{Y}\right \|_F & \leq \left \|\mathbf{R}\right \|_F \leq \left \|\mathbf{Y}\right \|_F. \tag*{$\blacksquare$}
	\end{align*}
\end{appendix}

\bibliographystyle{IEEEbib}
\bibliography{IEEEabrv,References}
\end{document}